\documentclass[12pt,a4paper]{article}
\usepackage{amssymb}
\usepackage{amsmath}
\usepackage{amsfonts}
\usepackage{geometry}
\usepackage{setspace}

\newtheorem{theorem}{Theorem}
\newtheorem{assumption}{Assumption}

\newtheorem{definition}{Definition}

\newtheorem{lemma}{Lemma}

\newtheorem{proposition}{Proposition}

\input{tcilatex}

\begin{document}

\title{Large sample properties of GMM estimators under second-order
identification.}
\author{Hugo Kruiniger\thanks{%
Address: hugo.kruiniger@durham.ac.uk; Department of Economics, 1 Mill Hill
Lane, Durham DH1 3LB, England.} \\
Durham University}
\date{This version: 3 August 2026}
\maketitle

\vspace{6.2cm}

\bigskip

\bigskip

\bigskip

\noindent JEL\ classification: C12, C13, C23.\bigskip

\noindent Keywords: bias, Generalized Method of Moments (GMM), moment
conditions, optimal weight matrix, rank deficiency, rate of convergence,
second-order local identification, underidentification.

\setcounter{page}{0} \thispagestyle{empty}

\newpage

\baselineskip=20pt

\renewcommand{\baselinestretch}{1.5}

\begin{center}
\textbf{Abstract}
\end{center}

\vspace{1cm}

Dovonon and Hall (Journal of Econometrics, 2018) proposed a limiting
distribution theory for GMM\ estimators for a $p$ - dimensional globally
identified parameter vector $\phi $ when local identification conditions
fail at first-order but hold at second-order. They assumed that the
first-order underidentification is due to the expected Jacobian having rank $%
p-1$ at the true value $\phi _{0}$, i.e., having a rank deficiency of one.
After reparametrizing the model such that the last column of the Jacobian
vanishes, they showed that the GMM\ estimator of the vector comprising the
first $p-1$ parameters, $\widehat{\phi }_{1},$ converges at rate $T^{-1/2}$
and the GMM estimator of the remaining parameter, $\widehat{\phi }_{p},$
converges at rate $T^{-1/4}$. They also provided a limiting distribution of $%
T^{1/4}(\widehat{\phi }_{p}-\phi _{0,p})$ subject to a (non-transparent)
condition which they claimed to be not restrictive in general. However, as
we show in this paper, their condition is in fact only satisfied when $\phi $
is overidentified and the limiting distribution of $T^{1/4}(\widehat{\phi }%
_{p}-\phi _{0,p}),$ which is non-standard, depends on whether $\phi $ is
exactly identified or overidentified. In particular, the limiting
distributions of the sign of $T^{1/4}(\widehat{\phi }_{p}-\phi _{0,p})$ for
the cases of exact and overidentification, respectively, are different and
are obtained by using expansions of the GMM objective function of different
orders. Unsurprisingly, we find that the limiting distribution theories of
Dovonon and Hall (2018) for Indirect Inference (II) estimation under two
different scenarios with second-order identification where the target
function is a GMM\ estimator of the auxiliary parameter vector, are
incomplete for similar reasons. We discuss how our results for GMM
estimation can be used to complete both theories and in particular how they
can be used to obtain the limiting distributions of the II estimators in the
case of exact identification under either scenario. We also derive the
optimal, in the sense of limiting Mean Squared Error minimising, weight
matrices for $\widehat{\phi }_{1}$ and $\widehat{\phi }_{p},$ respectively.

\setcounter{page}{0} \thispagestyle{empty}\newpage

\section{Introduction}

Global identification is a necessary condition for consistency of an
estimator. In models that are linear in the parameters, global
identification is equivalent to first-order local identification. However,
in models that are nonlinear in the parameters, global identication of the
parameter vector may hold even when some of the parameters are not
first-order but higher order locally identified although in this case the
rate of convergence of the estimators of these parameters is slower than the
usual rate.

For the situation where one of the elements of $\phi ,$ say $\phi _{p},$ is
not first-order but only second-order locally identified, Sargan (1983),
Rotnitzky et al. (2000) and Kruiniger (2013) developed asymptotic theory for
IV estimators, MLEs and Quasi MLEs, respectively. A common finding is that
the estimator of the parameter that is only second-order locally identified
converges at a quartic root rate, i.e., at rate $T^{-1/4}$ and has a
non-normal asymptotic distribution, while the estimators of the parameters
that are first-order locally identified converge at the usual square root
rate, i.e., at rate $T^{-1/2}$ and have asymptotic distributions that are
mixtures of normal distributions. Furthermore, the limiting distribution of $%
T^{1/2}(\widehat{\phi }_{p}-\phi _{0,p})^{2}$ is a mixture of a half-normal
distribution and $0$.

Dovonon and Renault (2009) give a formal definition of second-order local
identification in the context of GMM\ estimation. Dovonon and Hall (2018),
henceforth DH, present an asymptotic theory for GMM estimators under
second-order identification. The limiting distribution they give for the
estimator of the second-order locally identified parameter, i.e., $\phi _{p}$%
, holds if a certain condition is satisfied. However, as we show in this
paper, their condition is only satisfied when $\phi $ is overidentified.
Furthermore, we show that the limiting distribution of $\widehat{\phi }_{p}$
depends on whether $\phi $ is exactly identified or overidentified. In
particular, the limiting distributions of the sign of $T^{1/4}(\widehat{\phi 
}_{p}-\phi _{0,p})$ for the cases of exact and overidentification,
respectively, are different and are obtained by using expansions of the GMM
objective function of different orders. In fact, in the case of exact
identification the limiting distribution of the sign of $T^{1/4}(\widehat{%
\phi }_{p}-\phi _{0,p})$ may not even exist. The reason for these
differences is that in the case of exact identification some terms in the
expansion vanish. On the other hand, we find that the formula for the
limiting distribution of the GMM\ estimator of the vector with the other
elements of $\phi $, viz. $\widehat{\phi }_{1}$, is the same for both cases.

When the parameters are overidentified, the GMM estimator $\widehat{\phi }$
depends on a weight matrix. We show that the asymptotically optimal weight
matrix for $\widehat{\phi }_{1}$ is different from the asymptotically
optimal weight matrix for $\widehat{\phi }_{p}$, and that only the latter
weight matrix is the same as the one that is used when all the parameters
are first-order identified.

Kruiniger (2018a) derived the limiting distributions of two Modified MLEs
for a panel ARX(1) model with homoskedastic errors when the autoregressive
parameter equals one by viewing them as GMM\ estimators. In the unit root
case the parameter vector is only second-order locally identified by the
objective functions of the Modified MLEs due to the nonlinear terms in the
modified score vector. Alvarez and Arellano (2021) found that in the same
case (i.e., the case of a unit root and homoskedastic errors) the
autoregressive parameter of the panel AR(1) model is only second-order
locally identified by certain nonlinear moment conditions due to Ahn and
Schmidt (1995). DH showed that a set of moment conditions that are related
to a conditionally heteroskedastic factor model for asset returns has a rank
deficient Jacobian matrix and that the vector of parameters in these moment
conditions is second-order locally identified. Sargan (1983) discussed IV
and FIML estimation of dynamic simultaneous equation models that are linear
in the variables and nonlinear in the parameters and where the parameter
vector is only second-order locally identified. Finally, Rotnitzky et al.
(2000) give additional examples of models where the parameter vector is only
second-order locally identified.

The paper is organized as follows. Section 2 briefly reviews GMM estimation
under first-order local identification. Section 3 defines second-order
identification. Section 4 presents the limiting distribution theory for GMM
estimators under second-order local identification and discusses its
implications for Indirect Inference (II) estimation under two scenarios with
second-order local identification where the target function is a GMM\
estimator of the auxiliary parameter vector. Section 5 offers some
concluding remarks. The appendix contains the proofs.

\section{GMM under first-order identification$\protect\vspace{-0.09in}$}

In this section we briefly review the basic GMM framework based on
first-order asymptotics, paying special attention to the role of first-order
local identification. We first define the GMM estimator and then discuss
some first-order asymptotic theory for this estimator. To this end, we
introduce the following notation. The model involves the random vector $X$
which is assumed strictly stationary with distribution $P(\phi _{0})$ which
is indexed by the parameter vector $\phi \in \Phi \subset \mathcal{%
\mathbb{R}
}^{p}$. $\phi _{0}$\ is the true value of $\phi .$

GMM is a semi-parametric method in the sense that its implementation does
not require complete knowledge of $P($\textperiodcentered $)$ but only
population moment conditions implied by this distribution. In view of this,
we suppose that the model implies:%
\begin{equation}
E[g(X,\phi _{0})]=0,  \label{f1}
\end{equation}%
where $g($\textperiodcentered $)$ is a $q\times 1$ vector of continuous
functions. The GMM estimator of $\phi _{0}$ based on\ (1) is defined as:%
\begin{equation}
\widehat{\phi }=\underset{\phi \in \Phi }{argmin}Q_{T}(\phi ),
\end{equation}%
where%
\begin{equation*}
Q_{T}(\phi )=m_{T}^{\prime }(\phi )W_{T}m_{T}(\phi )\text{ with }m_{T}(\phi
)=T^{-1}\tsum\limits_{t=1}^{T}g(x_{t},\phi ),
\end{equation*}%
$W_{T}$ is a positive definite matrix, and $\{x_{t}\}_{t=1}^{T}$ represents
the sample observations on $X$.

We will assume that $q\geq p$ and that $m_{T}(\phi )$ satisfies

\begin{assumption}
(i) $m_{T}(\phi )=O_{p}(1)$ for all $\phi \in \Phi $; (ii) $%
T^{1/2}m_{T}(\phi _{0})\overset{d}{\rightarrow }N(0,V_{m})$, where $V_{m}$
is a positive definite matrix of finite constants.
\end{assumption}

To consider the first-order asymptotic properties of GMM estimators, we
introduce a number of high level assumptions.

\begin{assumption}
(i) $W_{T}\overset{p}{\rightarrow }W$, a positive definite matrix of
constants; (ii) $\Phi $ is a compact set; (iii) $Q_{T}(\phi )\overset{p}{%
\rightarrow }Q(\phi )=m(\phi )^{\prime }Wm(\phi )$ uniformly in $\phi $;
(iv) $Q(\phi )$ is continuous on $\Phi $; (v) $Q(\phi _{0})<Q(\phi )$ $%
\forall \phi \neq \phi _{0},$ $\phi \in \Phi $.
\end{assumption}

Assumption 2(v) serves as a global identification condition. These
conditions are sufficient to establish consistency, see, for example, Newey
and McFadden (1994).

\begin{proposition}
If Assumption 2 holds, then $\widehat{\phi }\overset{p}{\rightarrow }\phi
_{0}$.
\end{proposition}

Let $M_{T}(\tilde{\phi})=\partial m_{T}(\phi )/\partial \phi ^{\prime
}|_{\phi =\tilde{\phi}}$ and let $N_{\phi ,\epsilon }$ be an $\epsilon $%
-neighbourhood of $\phi _{0}$, that is,\linebreak $N_{\phi ,\epsilon
}=\{\phi :\parallel \phi -\phi _{0}\parallel <\epsilon \}$. We can derive
the first-order asymptotic distribution of $\widehat{\phi }$ after adding
the following assumption, cf. Newey and McFadden (1994).

\begin{assumption}
(i) $\phi _{0}$ is an interior point of $\Phi $; (ii) $m_{T}(\phi )$ is
continuously differentiable on $N_{\phi ,\epsilon }$; (iii) $M_{T}(\phi )%
\overset{p}{\rightarrow }M(\phi )$ uniformly on $N_{\phi ,\epsilon }$; (iv) $%
M(\phi )$ is continuous at $\phi _{0}$; (v) $M(\phi _{0})$ has rank $p$.
\end{assumption}

Assumption 3(v) is the condition for first-order local identification. It is
sufficient but not necessary for local identification of $\phi _{0}$ on $%
N_{\phi ,\epsilon }$, but it is necessary for the development of the
standard first-order asymptotic theory.

\begin{proposition}
If Assumptions 1--3 hold, then $T^{1/2}(\widehat{\phi }_{MD}-\phi _{0})%
\overset{d}{\rightarrow }N(0,V_{\phi })$, where

$V_{\phi }=[M(\phi _{0})^{\prime }WM(\phi _{0})]^{-1}M(\phi _{0})^{\prime
}WV_{m}WM(\phi _{0})[M(\phi _{0})^{\prime }WM(\phi _{0})]^{-1}$.
\end{proposition}

Global identification is crucial for consistency; global and first-order
local identification are needed for the preceding asymptotic distribution
theory.

Given Assumption 2(i), the global identification condition for GMM can be
equivalently stated as $E[g(X,\phi )]=0$ has a unique solution at $\phi
=\phi _{0}$. The first-order local identification condition can also be
stated as $E[\partial g(X,\phi )/\partial \phi ^{\prime }|_{\phi =\phi
_{0}}] $ has full column rank.

\section{Second-order local identification}

For our analysis of GMM, we adopt the definition of second-order local
identification originally introduced by Dovonon and Renault (2009). To
present this definition, we introduce the following notations. Let $m(\phi
)=E[g(X,\phi )]$ and%
\begin{equation*}
M_{k}^{(2)}(\phi _{0})=E\left[ \frac{\partial ^{2}g_{k}(X,\phi )}{\partial
\phi \partial \phi ^{\prime }}|_{\phi =\phi _{0}}\right] ,\text{ }%
k=1,2,\ldots ,q,
\end{equation*}%
where $g_{k}(X,\phi )$ is the $k$th element of $g(X,\phi )$ and $g($%
\textperiodcentered $)$ is defined in (\ref{f1}). Second-order local
identification is defined as follows.

\begin{definition}
The moment condition $m(\phi )=0$ locally identifies $\phi _{0}\in \Phi $ up
to the second order if:

(a) $m(\phi _{0})=0.$

(b) For all $u$ in the range of $M(\phi _{0})^{\prime }$ and all $v$ in the
nullspace of $M(\phi _{0})$, we have:$\medskip $

$\left( M(\phi _{0})u+\left( v^{\prime }M_{k}^{(2)}(\phi _{0})v\right)
_{1\leq k\leq q}=0\right) \Rightarrow (u=v=0).$
\end{definition}

The latter condition is derived using a second-order expansion of $m(\phi )$
around $m(\phi _{0})$ and can be motivated as follows. For any non-zero $%
\phi -\phi _{0}$ with $\phi \in N_{\phi ,\epsilon },$ we have $\phi -\phi
_{0}=c_{1}u+c_{2}v$ where $c_{1},$ $c_{2}$ are constants such that $%
c_{1}\neq 0$ and/or $c_{2}\neq 0$. For those directions for which $c_{1}$ is
non-zero, the first-order term is non-zero and dominates, and for those
directions in which $c_{1}=0$, the second-order term is non-zero. Thus,
without requiring the expected Jacobian matrix $M(\phi _{0})$ to have full
rank, conditions (a) and (b) in Definition 1 guarantee local identification
in the sense that there is no sequence of points $\{\phi _{n}\}$ different
from $\phi _{0}$ but converging to $\phi _{0}$ such that $m(\phi _{n})=0$
for all $n$. The difference between first-order local identification and
second-order local identification (with $M(\phi _{0})$ rank deficient) is
how sharply $m(\phi )$ moves away from $0$ in the neighbourhood of $\phi
_{0} $.

\noindent \textbf{Example}\textit{.} Consider the following panel AR(1)
model with individual effects:%
\begin{eqnarray}
y_{i,t} &=&\rho y_{i,t-1}+w_{i,t},  \label{mdl} \\
w_{i,t} &=&\eta _{i}+\varepsilon _{i,t},\text{ where }\eta _{i}=(1-\rho )\mu
_{i},  \notag
\end{eqnarray}%
for $i=1,...,N$ and $t=1,2,3.$ The number of individuals, $N,$ may be large.
Note that\ when $\rho =1,$ then $\eta _{i}=0$. Adding the term $%
x_{i,t}^{\prime }\breve{\beta}$ to (\ref{mdl}) with $x_{i,t}$ exogenous and $%
\breve{\beta}=\beta (1-\rho )$ does not affect the essence of the analysis
below except that $p$ and $q$ increase by $\dim (\breve{\beta}$)\ and $%
\breve{\beta}$ is first-order locally identified in all cases. We make the
following assumption.\smallskip

\textbf{Assumption A} $\left\{ \eta
_{i},y_{i,0},y_{i,1},y_{i,2},y_{i,3}\right\} _{i=1}^{N}$ \textit{is a random
sample from a joint distribution with finite fourth-order moments that
satisfies} $E(\varepsilon _{i,t}|\eta _{i},y_{i,0},\ldots ,y_{i,t-1})=0$ 
\textit{for} $t=1,2,3.$\medskip

Let the unconditional variances of the errors be denoted as $E(\varepsilon
_{i,t}^{2})=\sigma _{t}^{2}$ for $t=1,2,3.$ We are interested in GMM
estimation of $\rho $. Assumption A implies the following three linear
moment conditions for the model in (\ref{mdl}), cf. Arellano and Bond (1991):%
\begin{equation}
m_{AB,s,t}(\rho ):=E[y_{i,t-s}(\Delta y_{i,t}-\rho \Delta y_{i,t-1})]=0\text{
for }s=2,t\text{ and }t=2,3,  \label{bega}
\end{equation}%
where $\Delta y_{i,t}=y_{i,t}-y_{i,t-1}$. Assumption A also implies one
nonlinear moment condition\ for the model in (\ref{mdl}), cf. Ahn and
Schmidt (1995):%
\begin{equation}
m_{AS,3}(\rho ):=E[(y_{i,3}-\rho y_{i,2})(\Delta y_{i,2}-\rho \Delta
y_{i,1})]=0.  \label{asm}
\end{equation}%
\pagebreak

If $\rho \neq 1,$ then $\rho $ is both globally and first-order locally
identified by each of the above four moment conditions, while if $\rho =1,$
then the first three moment conditions do not help to identify $\rho $ at
all, because they are linear in $\rho $ and $dm_{AB,s,t}(\rho )/d\rho =0$
for $s=2,t$ and $t=2,3.$ When $\rho =1$ and $\sigma _{1}^{2}\neq \sigma
_{2}^{2}$, then $\rho $ is still first-order locally identified by $%
m_{AS,3}(\rho )=0,$ because $dm_{AS,3}(\rho )/d\rho =-\sigma _{2}^{2}+\sigma
_{1}^{2}\neq 0,$ but \textit{technically speaking} $\rho $ is no longer
globally identified because $m_{AS,3}(\rho )=0$ now has two solutions, i.e., 
$\rho =1$ and $\rho =\sigma _{2}^{2}/\sigma _{1}^{2}$. However, the fact
that $m_{AS,3}(\rho )=0$ has multiple solutions only in this case (and not
when $\rho \neq 1$) means that \textit{in practice} $\rho $ is globally
identified in this case (because the occurance of two roots implies that $%
\rho =1$) and that the GMM\ estimator that exploits $m_{AS,3}(\rho )=0$ is
also consistent in this case, see Kruiniger (2013) for details. If $\rho =1$
and $\sigma _{1}^{2}=\sigma _{2}^{2}$, then $\rho $ is second-order rather
than first-order locally identified by $m_{AS,3}(\rho )=0,$ because $%
dm_{AS,3}(\rho )/d\rho =0$ and $d^{2}m_{AS,3}(\rho )/d\rho ^{2}=2\sigma
_{1}^{2}\neq 0$, and $\rho $ is also globally identified because the two
solutions of $m_{AS,3}(\rho )=0$ are now both equal to $1$, that is, $1$ is
a double root of $m_{AS,3}(\rho )=0$, cf. Alvarez and Arellano (2021).
Finally, as we have just seen, when $\rho =1,$ then $\rho $ is only
identified by $m_{AS,3}(\rho )=0$ and not by $m_{AB,s,t}(\rho )=0$ for $%
s=2,t $ and $t=2,3,$ so in this case we really have $q=1$ rather than $q=4,$
that is, $\rho $ is exactly identified rather than overidentified because $%
q=p=1.$

\section{The limiting distribution of the GMM estimator}

In this section we consider the moment condition model (1) and study the
asymptotic behaviour of the GMM estimator when $\phi _{0}$ is second-order
locally identified because the moment condition exhibits the properties in
Definition 1 but the standard local identification condition (Assumption
3(v)) fails.

\subsection{Main results}

We study the asymptotic behaviour of the GMM estimator by restricting
ourselves to the case of a rank deficiency of one, i.e., the rank of $M(\phi
_{0})$ is equal to $p-1$, since this case is relatively easy to analyse
compared to the general case. W.l.o.g. we consider the case where the rank
deficiency of $M(\phi _{0})$ is due to its last column being a null vector.%
\footnote{%
As mentioned by Sargan (1983), any model with an expected Jacobian that has
a rank deficiency of one can be brought into this configuration by
reparametrizing the model as needed.} To\ this end, we partition $\phi $
into ($\phi _{1:p-1}^{\prime },\phi _{p})^{\prime }$ where $\phi _{1:p-1}$
is the vector consisting of the first $p-1$ elements of $\phi $ and $\phi
_{p}$ is the $p-$th element of $\phi $. For ease of presentation below, we
shorten the subscript and write $\phi _{1}$ for $\phi _{1:p-1}$. Thus $\phi
_{0}=(\phi _{0,1}^{\prime },\phi _{0,p})^{\prime }$ where $\phi _{0,1}$ is a 
$(p-1)\times 1$ vector containing the true value of $\phi _{1:p-1}$ and $%
\phi _{0,p}$ is the true value of $\phi _{p}$. If $M(\phi _{0})$ has rank $%
p-1$ with $\frac{\partial m}{\partial \phi _{p}}(\phi _{0})=0$, then
second-order local identification is equivalent to:%
\begin{equation*}
Rank\left( \frac{\partial m}{\partial \phi _{1}^{\prime }}(\phi _{0})\text{%
\quad }\frac{\partial ^{2}m}{\partial \phi _{p}^{2}}(\phi _{0})\right) =p
\end{equation*}%
This is the setting studied by Sargan (1983) for the instrumental variables
estimator for a nonlinear in parameters model. We now present the regularity
conditions under which we derive the asymptotic distribution of the GMM
estimator. Define $D=\frac{\partial m}{\partial \phi _{1}^{\prime }}(\phi
_{0})$ and $G=$\linebreak $\frac{\partial ^{2}m}{\partial \phi _{p}^{2}}%
(\phi _{0})$. The next assumption states formally the identification pattern
described above.

\begin{assumption}
(i) $m(\phi )=0$ $\Leftrightarrow $ $\phi =\phi _{0}$; (ii) $\frac{\partial m%
}{\partial \phi _{p}}(\phi _{0})=0$; (iii) $Rank(D$ $G)=p$.
\end{assumption}

As mentioned by Sargan (1983), any model with an expected Jacobian that has
a rank deficiency of one can be brought into this configuration by rotating
the parameter space, see DH for details. We also require the following
regularity conditions to hold.

\begin{assumption}
(i) $m_{T}(\phi )$ has partial derivatives up to order 3 if $q>p$ and up to
order 5 if $q=p$ in a neighbourhood $N_{\phi ,\epsilon }$ of $\phi _{0}$ and
the derivatives of $m_{T}(\phi )$ converge in probability uniformly on $%
N_{\phi ,\epsilon }$ to those of $m(\phi )$.$\smallskip \smallskip
\smallskip \smallskip $

(ii) $\sqrt{T}\left( 
\begin{array}{c}
m_{T}(\phi _{0}) \\ 
\frac{\partial m_{T}}{\partial \phi _{p}}(\phi _{0})%
\end{array}%
\right) \overset{d}{\rightarrow }\left( 
\begin{array}{c}
\mathcal{Z}_{0} \\ 
\mathcal{Z}_{1}%
\end{array}%
\right) .\smallskip \smallskip \smallskip \smallskip $

(iii) $W_{T}-W=o_{p}(T^{-1/4}),$ $\frac{\partial m_{T}}{\partial \phi
_{1}^{\prime }}(\phi _{0})-D=O_{p}(T^{-1/2}),$ $\frac{\partial ^{2}m_{T}}{%
\partial \phi _{p}^{2}}(\phi _{0})-G=O_{p}(T^{-1/2}),$ $\frac{\partial
^{2}m_{T}}{\partial \phi _{1}^{\prime }\partial \phi _{p}}(\phi
_{0})-G_{1p}=o_{p}(1)$ and $\frac{\partial ^{3}m_{T}}{\partial \phi _{p}^{3}}%
(\phi _{0})-L=o_{p}(1),$ and if $q=p,$ $\frac{\partial ^{3}m_{T}}{\partial
\phi _{1}^{\prime }\partial \phi _{p}^{2}}(\phi _{0})-G_{1pp}=o_{p}(1),$ $%
\frac{\partial ^{4}m_{T}}{\partial \phi _{1}^{\prime }\partial \phi _{p}^{3}}%
(\phi _{0})-G_{1ppp}=o_{p}(1),$ $\frac{\partial ^{4}m_{T}}{\partial \phi
_{p}^{4}}(\phi _{0})-F=o_{p}(1)$ and $\frac{\partial ^{2}m_{k,T}}{\partial
\phi _{1}\partial \phi _{1}^{\prime }}(\phi _{0})-K_{k}=o_{p}(1)$ for $%
k=1,2,\ldots ,q,$ with $G_{1p}=\frac{\partial ^{2}m}{\partial \phi
_{1}^{\prime }\partial \phi _{p}}(\phi _{0}),$ $L=\frac{\partial ^{3}m}{%
\partial \phi _{p}^{3}}(\phi _{0}),$ $G_{1pp}=\frac{\partial ^{3}m}{\partial
\phi _{1}^{\prime }\partial \phi _{p}^{2}}(\phi _{0}),$ $G_{1ppp}=\frac{%
\partial ^{4}m}{\partial \phi _{1}^{\prime }\partial \phi _{p}^{3}}(\phi
_{0}),$ $F=\frac{\partial ^{4}m}{\partial \phi _{p}^{4}}(\phi _{0})$ and $%
K_{k}=\frac{\partial ^{2}m_{k}}{\partial \phi _{1}\partial \phi _{1}^{\prime
}}(\phi _{0})$ for $k=1,2,\ldots ,q,$ where $m_{k,T}(\phi )$\ $(m_{k}(\phi
)) $ is the kth element of $m_{T}(\phi )$\ $($of $m(\phi )).\vspace{-0.05in}$
\end{assumption}

These conditions are stronger than those imposed in the standard first-order
asymptotic analysis. The derivation of the asymptotic distribution of the
GMM estimator requires an expansion of $Q_{T}(\phi )$ involving derivatives
of $m_{T}(\phi )$ up to the third order when $q>p$ and up to the fifth order
when $q=p,$ and the uniform convergence guaranteed by Assumption 5(i) is
useful to control the remainder of our expansions. The orders of our
expansions of $Q_{T}(\phi )$ for the cases $q>p$ and $q=p$ are different
because in the case of exact identification some terms in the expansion
vanish. Assumption 5(ii) states that $\sqrt{T}(m_{T}(\phi _{0})^{\prime
},\partial m_{T}(\phi _{0})^{\prime }/\partial \phi _{p})^{\prime }$
converges in distribution. Under Assumption 4 and additional mild conditions
on $g(X,\phi _{0})$ and $\frac{\partial g}{\partial \phi _{p}}(X,\phi _{0})$%
, the central limit theorem guarantees that $(\mathcal{Z}_{0},\mathcal{Z}%
_{1})^{\prime }\sim N(0,v)$, with $v=lim_{T\rightarrow \infty }Var[\sqrt{T}%
(m_{T}(\phi _{0})^{\prime },\partial m_{T}(\phi _{0})^{\prime }/\partial
\phi _{p})^{\prime }]$. Assumption 5(iii) imposes the asymptotic order of
magnitude on the differences between some sample dependent quantities and
their probability limits. These orders of magnitude are enough to make these
differences negligible in the expansions. Assumption 5(iii) is not
particularly restrictive since most of the orders of magnitude imposed are
guaranteed by the central limit theorem.

To facilitate the presentation of our main result in this section, we
introduce the following definitions. Let $M_{d}$ be the matrix of the
orthogonal projection on the orthogonal complement of $W^{1/2}D:\vspace{%
-0.14in}$%
\begin{equation*}
M_{d}=I_{q}-W^{1/2}D(D^{\prime }WD)^{-1}D^{\prime }W^{1/2},\vspace{-0.14in}
\end{equation*}%
where $I_{q}$ is the identity matrix of size $q$, let $P_{g}$ be the matrix
of the orthogonal projection\ on $M_{d}W^{1/2}G:\vspace{-0.06in}$%
\begin{equation*}
P_{g}=M_{d}W^{1/2}G(G^{\prime }W^{1/2}M_{d}W^{1/2}G)^{-1}G^{\prime
}W^{1/2}M_{d},\vspace{-0.1in}
\end{equation*}%
and let $M_{dg}$ be the matrix of the orthogonal projection on the
orthogonal complement of $(W^{1/2}D$\quad $W^{1/2}G):\vspace{-0.06in}$%
\begin{equation*}
M_{dg}=M_{d}-P_{g}.\vspace{-0.1in}
\end{equation*}%
Let$\vspace{-0.14in}$%
\begin{eqnarray}
\mathcal{R}_{1} &=&(\mathcal{Z}_{0}^{\prime }W^{1/2}P_{g}W^{1/2}\mathcal{Z}%
_{0}G^{\prime }-G^{\prime }W^{1/2}P_{g}W^{1/2}\mathcal{Z}_{0}\mathcal{Z}%
_{0}^{\prime })W^{1/2}M_{d}W^{1/2}\times \vspace{-0.14in}  \notag \\
&&(\frac{1}{3}L+G_{1p}HG)/\sigma _{G}+\mathcal{Z}_{0}^{\prime
}W^{1/2}M_{dg}W^{1/2}(\mathcal{Z}_{1}+G_{1p}H\mathcal{Z}_{0}),  \label{f7}
\end{eqnarray}%
with $\sigma _{G}=G^{\prime }W^{1/2}M_{d}W^{1/2}G$ and $H=-(D^{\prime
}WD)^{-1}D^{\prime }W$. In addition, let $V=$ $-2\mathcal{Z}\mathbf{1}(%
\mathcal{Z}<0)/\sigma _{G},$ where $\mathcal{Z}=G^{\prime
}W^{1/2}M_{d}W^{1/2}\mathcal{Z}_{0}$ and $\mathbf{1}($\textperiodcentered $)$
is the usual indicator function.

The following lemma, which is based on Theorem 1 in DH, and theorem give the
asymptotic properties of the GMM estimator $\widehat{\phi }$ under
Assumptions 2, 4 and 5.\pagebreak

\begin{lemma}[Dovonon and Hall (2018)]
Under Assumptions 2, 4 and 5, we have:

$(a)$ $\widehat{\phi }_{1}-\phi _{0,1}=O_{p}(T^{-1/2})$ and $\widehat{\phi }%
_{p}-\phi _{0,p}=O_{p}(T^{-1/4});$

$(b)$ if in addition $\phi _{0}\in interior(\Phi )$, then$\smallskip
\smallskip $

$\qquad \left( 
\begin{array}{c}
\sqrt{T}(\widehat{\phi }_{1}-\phi _{0,1}) \\ 
\sqrt{T}(\widehat{\phi }_{p}-\phi _{0,p})^{2}%
\end{array}%
\right) \overset{d}{\rightarrow }\left( 
\begin{array}{c}
H\mathcal{Z}_{0}+HGV/2 \\ 
V%
\end{array}%
\right) .$
\end{lemma}

\begin{theorem}
Under Assumptions 2, 4 and 5, and if $\phi _{0}\in interior(\Phi )$, we have:

$(a)$ if in addition $q>p$, then$\smallskip \smallskip $

$\qquad T^{1/4}(\widehat{\phi }_{p}-\phi _{0,p})\overset{d}{\rightarrow }%
(-1)^{B_{1}}\sqrt{V},\smallskip \smallskip $

with $B_{1}=\mathbf{1}(\mathcal{R}_{1}\geq 0);$

$(b)$ if in addition $q=p$ and $\Pr (\mathcal{R}_{2}=0|\mathcal{Z}<0)=0$,
where $\mathcal{R}_{2}$ is defined in the proof below equation $(\ref{h4})$,
then$\smallskip \smallskip $

$\qquad T^{1/4}(\widehat{\phi }_{p}-\phi _{0,p})\overset{d}{\rightarrow }%
(-1)^{B_{2}}\sqrt{V},\smallskip \smallskip $

with $B_{2}=\mathbf{1}(\mathcal{R}_{2}\geq 0).$
\end{theorem}

Parts (a) and (b) of Lemma 1 are the same as parts (a) and (b) of Theorem 1
in DH. In the Appendix we provide an alternative, self-contained proof for
part (b) of Lemma 1. There we also provide a proof for our Theorem 1.$%
\vspace{-0.14in}$

\subsection{Discussion$\protect\vspace{-0.06in}$}

Our theorem is different from part (c) of Theorem 1 of DH. The latter only
states that $T^{1/4}(\widehat{\phi }_{p}-\phi _{0,p})$ converges in
distribution to the limiting distribution given in part (a) of our Theorem 1
under the condition that $\Pr (\mathcal{R}_{1}=0)=0.$ DH claim in their
Remark 1 that this condition \textquotedblleft is not expected to be
restrictive in general ...\textquotedblright , although they also add the
following caveat: \textquotedblleft However, when $q=p=1$ (one moment
restriction with one non first-order locally identified parameter), we can
see that $\mathcal{R}_{1}=0$.\textquotedblright\ However, as we show in the
proof of Lemma 2 in the Appendix, the condition $\Pr (\mathcal{R}_{1}=0)=0$
is actually only satisfied when $\phi $ is overidentified, i.e., when $q>p;$
when $\phi $ is exactly identified, i.e., when $q=p,$ then $\Pr (\mathcal{R}%
_{1}=0|\mathcal{Z}<0)=1$ and hence $\Pr (\mathcal{R}_{1}=0)>0.$ In other
words, the theory of DH only provides the limiting distribution of $T^{1/4}(%
\widehat{\phi }_{p}-\phi _{0,p})$ for the case where $\phi $ is
overidentified.

In part (b) of Theorem 1 we provide the limiting distribution of $T^{1/4}(%
\widehat{\phi }_{p}-\phi _{0,p})$ for the case where $\phi $ is exactly
identified and $\Pr (\mathcal{R}_{2}=0|\mathcal{Z}<0)=0$. The difference
between the limiting distributions of $T^{1/4}(\widehat{\phi }_{p}-\phi
_{0,p})$ given in parts (a) and (b) is related to the difference between the
distributions of the Bernoulli r.v.'s $B_{1}$ and $B_{2}$ that determine the
sign of $T^{1/4}(\widehat{\phi }_{p}-\phi _{0,p})$ when $\mathcal{Z}<0.$

Let $S_{1,T}=O_{p}(1)$ and $S_{2,T}=O_{p}(1)$ be such that $\sqrt{T}%
m_{T}(\phi _{0})+D\sqrt{T}(\widehat{\phi }_{1}-\phi _{0,1})+\frac{1}{2}G%
\sqrt{T}(\widehat{\phi }_{p}-\phi _{0,p})^{2}=(\widehat{\phi }_{p}-\phi
_{0,p})S_{1,T}+T^{-1/2}S_{2,T}+o_{p}(T^{-1/2})$, $S_{1,T}\overset{d}{%
\rightarrow }S_{1}$ and $S_{2,T}\overset{d}{\rightarrow }S_{2}$. Then the
condition $\Pr (\mathcal{R}_{2}=0|\mathcal{Z}<0)=0$ is violated if and only
if $F=0,$ $G_{1pp}=0,$ $K_{k}=0$ for $k=1,2,\ldots ,q,$ and $S_{2}=0$, cf.
Lemma 2(b). In particular, when $q=p=1,$ then $\Pr (\mathcal{R}_{2}=0|%
\mathcal{Z}<0)>0$ if only if $F=0$ and $S_{2}=0$. These conditions for $\Pr (%
\mathcal{R}_{2}=0|\mathcal{Z}<0)>0,$ including the condition $S_{2}=0$,
require that $m_{T}(\phi )$ is at most locally linear in $\phi _{1}$ and at
most locally cubic in $\phi _{p}$ at $\phi _{0}$.\ If $q=p$ and $\Pr (%
\mathcal{R}_{2}=0|\mathcal{Z}<0)>0,$ then it is not possible to describe the
limiting distribution of the sign of $T^{1/4}(\widehat{\phi }_{p}-\phi
_{0,p})$.

Part (a) of Lemma 1 gives the rates of convergence of $\widehat{\phi }_{1}$
and $\widehat{\phi }_{p}$. Because $\phi _{1}$ is first-order identified and 
$\phi _{p}$ is second-order identified, $\widehat{\phi }_{1}-\phi _{0,1}$
converges at the usual rate $T^{-1/2}$ while $\widehat{\phi }_{p}-\phi
_{0,p} $ converges at the slower rate $T^{-1/4}.$

Part (b) of Lemma 1 gives the limiting distribution of $(\sqrt{T}(\widehat{%
\phi }_{1}-\phi _{0,1}),$ $\sqrt{T}(\widehat{\phi }_{p}-\phi _{0,p})^{2})$.
This result is obtained by minimizing the sum of the leading $O_{p}(T^{-1})$
terms of an expansion of $m_{T}^{\prime }(\widehat{\phi })W_{T}m_{T}(%
\widehat{\phi })$ around $\phi _{0}$ which are collected into $\underline{K}%
_{T}(\phi _{0})$ as given by (\ref{g4}) in the Appendix. As $\underline{K}%
_{T}(\phi _{0})$ is a quadratic function of $(\widehat{\phi }_{1}-\phi
_{0,1})$ and $(\widehat{\phi }_{p}-\phi _{0,p})^{2}$ only, it only allows
one to obtain the limiting distribution of $T^{1/4}|\widehat{\phi }_{p}-\phi
_{0,p}|.$ To obtain the limiting distribution of the sign of $T^{1/4}(%
\widehat{\phi }_{p}-\phi _{0,p})$ one needs to employ a higher order
expansion of $m_{T}^{\prime }(\widehat{\phi })W_{T}m_{T}(\widehat{\phi })$
which includes an odd power of $(\widehat{\phi }_{p}-\phi _{0,p}).$

When $q>p,$ the limiting distribution of the sign of $T^{1/4}(\widehat{\phi }%
_{p}-\phi _{0,p})$ can be obtained from the $O_{p}(T^{-5/4})$ terms in the
expansion of $m_{T}^{\prime }(\widehat{\phi })W_{T}m_{T}(\widehat{\phi }).$
In that case we have:$\vspace{-0.13in}$%
\begin{equation*}
m_{T}^{\prime }(\widehat{\phi })W_{T}m_{T}(\widehat{\phi })=\underline{K}%
_{T}(\phi _{0,p})+(\widehat{\phi }_{p}-\phi _{0,p})\times
2R_{1T}+o_{p}(T^{-5/4}),\vspace*{-0.11in}
\end{equation*}%
where $\underline{K}_{T}(\phi _{0,p})$ and $R_{1T}$ are quadratic functions
of $(\widehat{\phi }_{p}-\phi _{0,p})^{2}.$ We show in the Appendix that $%
TR_{1T}\overset{d}{\rightarrow }\mathcal{R}_{1}$ and that $\Pr (\mathcal{R}%
_{1}=0)=0$ if $q>p.$ When $Z_{T}\equiv G^{\prime
}W^{1/2}M_{d}W^{1/2}m_{T}(\phi _{0})<0,$ $q>p$ and $T$ is large, the minimum
of $m_{T}^{\prime }(\phi )W_{T}m_{T}(\phi )$ is reached when $(\phi
_{p}-\phi _{0,p})$ has the opposite sign to $R_{1T}$. This suggests that
when $q>p,$ then the limiting distribution of the sign of $T^{1/4}(\widehat{%
\phi }_{p}-\phi _{0,p})$ can be described by $(-1)^{B_{1}}$ with $B_{1}=%
\mathbf{1}(\mathcal{R}_{1}\geq 0).$

When $q=p,$ then $\Pr (\mathcal{R}_{1}=0|\mathcal{Z}<0\mathbf{)}=1$. In
fact, when $q=p,$ then $\mathcal{R}_{1}=0$, see the end of the proof of
Lemma 2(a1). However, if $q=p$ and $\Pr (\mathcal{R}_{2}=0|\mathcal{Z}<0)=0,$
then the limiting distribution of the sign of $T^{1/4}(\widehat{\phi }%
_{p}-\phi _{0,p})$ can be obtained from the $O_{p}(T^{-7/4})$ terms in the
expansion of $m_{T}^{\prime }(\widehat{\phi })W_{T}m_{T}(\widehat{\phi }).$
In that case we have:$\vspace{-0.12in}$%
\begin{equation*}
m_{T}^{\prime }(\widehat{\phi })W_{T}m_{T}(\widehat{\phi })=\underline{K}%
_{T}(\phi _{0,p})+(\widehat{\phi }_{p}-\phi _{0,p})\times
2R_{1T}+O_{p}(T^{-6/4})+(\widehat{\phi }_{p}-\phi _{0,p})\times
2R_{2T}+o_{p}(T^{-7/4}),\vspace*{-0.1in}
\end{equation*}%
where the $O_{p}(T^{-6/4})$ term and $R_{2T}$ are cubic functions of $(%
\widehat{\phi }_{p}-\phi _{0,p})^{2}.$ We show in the Appendix that $%
T^{6/4}R_{2T}\overset{d}{\rightarrow }\mathcal{R}_{2}$ if $q=p.$ When $%
Z_{T}<0$, $q=p,$ $\Pr (\mathcal{R}_{2}=0|\mathcal{Z}<0)=0$ and $T$ is large,
then the minimum of $m_{T}^{\prime }(\phi )W_{T}m_{T}(\phi )$ is reached
when $(\phi _{p}-\phi _{0,p})$ has the opposite sign to $R_{2T}$. This
suggests that if $q=p$ and $\Pr (\mathcal{R}_{2}=0|\mathcal{Z}<0)=0$, then
the limiting distribution of the sign of $T^{1/4}(\widehat{\phi }_{p}-\phi
_{0,p})$ can be described by $(-1)^{B_{2}}$ with $B_{2}=\mathbf{1}(\mathcal{R%
}_{2}\geq 0).$ If $q=p$ and $\Pr (\mathcal{R}_{2}=0|\mathcal{Z}<0)>0,$ then
the sign of $T^{1/4}(\widehat{\phi }_{p}-\phi _{0,p})$ does not have a
proper limiting distribution, whereas $\sqrt{T}(\widehat{\phi }_{p}-\phi
_{0,p})^{2}$ has one, which is given in Lemma 1(b).

The limiting distributions in Theorem 1 and part (b) of Lemma 1 are
non-standard but easy to simulate. Approximations to these limiting
distributions can be obtained by drawing randomly copies of $(\mathcal{Z}%
_{0}^{\prime },$ $\mathcal{Z}_{1}^{\prime })^{\prime }$ from $N(0,\widehat{v}%
)$, where $\widehat{v}$ is a consistent estimator of $v$, and using
consistent estimators of $W,$ $D,$ $G,$ $L,$ $G_{1p},$ $G_{1pp},$ $G_{1ppp},$
$F$ and $K_{k}$ for $k=1,2,\ldots ,q$ as required.

Generally, the limiting distribution of\ $T^{1/4}(\widehat{\phi }_{p}-\phi
_{0,p})$ given in Theorem 1 is asymmetric around $0$ and $\widehat{\phi }%
_{p} $ has an asymptotic bias, unless $\Pr (B_{1}=1|\mathcal{Z}_{0})=\frac{1%
}{2}$ when $q>p$ or $\Pr (B_{2}=1|\mathcal{Z}_{0})=\frac{1}{2}$ when $q=p.$
Furthermore, $\widehat{\phi }_{1}$ has an asymptotic bias.\vspace{-0.12in}

\subsubsection{Optimal weight matrices for $\protect\widehat{\protect\phi }%
_{1}$ and $\protect\widehat{\protect\phi }_{p}$}

When $q>p,$ the limiting Mean Squared Errors (MSEs) of $\widehat{\phi }_{1}$
and $\widehat{\phi }_{p}$ depend on the choice of the limiting weight matrix 
$W$. Recall that $V$ depends on $M_{d}$ and that $M_{d}$ depends on $W$.
Hence $V$ depends on $W.$ Let $M_{d}(W)=M_{d},$ $V(W)=V,$ $\Psi _{1}(W)=E((%
\mathcal{Z}_{0}+GV(W)/2)(\mathcal{Z}_{0}+GV(W)/2)^{\prime })$ and $\Psi
_{p}=E(\mathcal{Z}_{0}\mathcal{Z}_{0}^{\prime }).$ Furthermore, let $\hat{%
\Psi}_{p}$ be a consistent estimate of $\Psi _{p}.$ At the end of the
appendix we show that the optimal weight matrix for $\widehat{\phi }_{p}$
that minimizes its limiting MSE is $W_{p,opt}=\Psi _{p}^{-1}$ and that the
limiting MSE of $T^{1/4}(\widehat{\phi }_{p,opt}-\phi _{0,p})$ is given by $%
(G^{\prime }\Psi _{p}^{-1/2}M_{d}(\Psi _{p}^{-1})\Psi _{p}^{-1/2}G)^{-1/2}%
\sqrt{2/\pi }.$ If $\phi _{1}$ and $\phi _{p}$ are estimated jointly, then
the weight matrix that minimizes the limiting MSE of $\widehat{\phi }_{1}$,
viz. $W_{opt},$ is the solution of $W_{opt}=(\Psi _{1}(W_{opt}))^{-1}$.
However, as $W_{opt}\neq c\Psi _{p}^{-1}$ for any $c\in 
\mathbb{R}
_{+}$, $W_{opt}$ is not an optimal weight matrix for $\widehat{\phi }_{p}.$
Moreover, $W_{opt}=(\Psi _{1}(W_{opt}))^{-1}$ has no closed form solution.
Therefore, it is better to estimate $\phi _{1}$ and $\phi _{p}$ separately.
In that case $W_{1,opt}=(\Psi _{1}(\Psi _{p}^{-1}))^{-1}$ is an optimal
weight matrix for $\widehat{\phi }_{1}$. Let $\widehat{W}_{1,opt}$ denote a
consistent estimate of $W_{1,opt}=(\Psi _{1}(\Psi _{p}^{-1}))^{-1}$. Then we
obtain $\widehat{\phi }_{1,opt}$ by minimising $m_{T}^{\prime }(\phi _{1};%
\widehat{\phi }_{p,opt})\widehat{W}_{1,opt}m_{T}(\phi _{1};\widehat{\phi }%
_{p,opt}).$ $\widehat{W}_{1,opt}$ can be obtained by making use of the
equality $E((\mathcal{Z}_{0}+GV(W)/2)(\mathcal{Z}_{0}+GV(W)/2)^{\prime })=%
\frac{1}{2}E((\mathcal{Z}_{0}-G\mathcal{Z}/\sigma _{G})(\mathcal{Z}_{0}-G%
\mathcal{Z}/\sigma _{G})^{\prime })+\frac{1}{2}E(\mathcal{Z}_{0}\mathcal{Z}%
_{0}^{\prime })$ with $W=\Psi _{p}^{-1}$ and taking the inverse of a
consistent estimate of its right-hand side.

\subsection{Examples with exact identification}

Kruiniger (2018a) derived the limiting distributions of two Modified MLEs
for the panel ARX(1) model with homoskedastic errors when the autoregressive
parameter equals one by viewing them as GMM\ estimators. In the unit root
case the autoregressive parameter is only second-order locally identified by
the objective functions of both Modified MLEs due to the nonlinear terms in
the modified score vector. Furthermore, the parameter vector is obviously
exactly identified and the condition $\Pr (\mathcal{R}_{2}=0)=0$ holds
because the condition $F+3!G_{1pp}HG+\frac{4!}{2}\tilde{\lambda}_{3}\neq 0$
is satisfied. It is therefore unsurprising that the limiting distributions
obtained by Kruiniger (2018a) for both Modified MLEs of the autoregressive
parameter are in agreement with Theorem 1(b) above.

We now return to the example given in section 3. Recall that when $\rho =1$
and $\sigma _{1}^{2}=\sigma _{2}^{2}$, then $\rho $ is second-order globally
identified by the four moment conditions mentioned in that example. However,
in this example $q=p=1,$ $F=0$ and $S_{2}=0$ because there is only one
nonlinear moment condition that identifies $\rho $, namely $m_{AS,3}(\rho
)=0 $, which is quadratic in $\rho $. Thus in this case the limiting
distribution of $T^{1/4}(\widehat{\rho }-1)$ cannot be obtained from Theorem
1(b). This is not due to a shortcoming of the theory; there are at least two
explanations for this result. Firstly, note that when $\rho =1,$ $%
m_{AS,3}(r)=E[(\varepsilon _{i,3}+(1-r)y_{i,2})(\Delta \varepsilon
_{i,2}+(1-r)\varepsilon _{i,1})].$ Hence, when $N$ is large, the GMM
estimator for $\rho $ is approximately equal to the solution of $%
N^{-1/2}\tsum_{i=1}^{N}[(\widehat{\rho }-1)^{2}y_{i,2}\varepsilon _{i,1}-(%
\widehat{\rho }-1)(\varepsilon _{i,1}\varepsilon _{i,3}+y_{i,2}\Delta
\varepsilon _{i,2})+\varepsilon _{i,3}\Delta \varepsilon _{i,2}]=0.$
However, the linear term $-(\widehat{\rho }-1)N^{-1/2}\tsum_{i=1}^{N}(%
\varepsilon _{i,1}\varepsilon _{i,3}+y_{i,2}\Delta \varepsilon
_{i,2})=o_{p}(1)$ because $N^{1/4}(\widehat{\rho }-1)=O_{p}(1)$ and $%
N^{-1/2}\tsum_{i=1}^{N}(\varepsilon _{i,1}\varepsilon _{i,3}+y_{i,2}\Delta
\varepsilon _{i,2})=O_{p}(1).$ Thus when $N$ tends to infinity, the sample
counterpart of $m_{AS,3}(\rho )=0$ determines the distribution of $N^{1/2}(%
\widehat{\rho }-1)^{2},$ i.e., $N^{1/2}(\widehat{\rho }-1)^{2}\overset{d}{%
\rightarrow }-\widetilde{\mathcal{Z}}_{0}\mathbf{1}(\widetilde{\mathcal{Z}}%
_{0}<0)/\sigma ^{2}$ with $N^{-1/2}\tsum_{i=1}^{N}(\varepsilon _{i,3}\Delta
\varepsilon _{i,2})\overset{d}{\rightarrow }\widetilde{\mathcal{Z}}_{0},$
but cannot determine the distribution of the sign of $N^{1/4}(\widehat{\rho }%
-1)$. Secondly, if $\rho =1$ and $\sigma _{1}^{2}=\sigma _{2}^{2}$, then
both roots of $m_{AS,3}(r)=0$ are equal to $1$, and hence \emph{both} roots
of the sample counterpart of $m_{AS,3}(r)=0$, viz. $\widehat{\rho }_{1}$ and 
$\widehat{\rho }_{2}$, are consistent estimators of $\rho $. The limiting
distributions of $N^{1/2}(\widehat{\rho }_{1}-1)^{2}$ and $N^{1/2}(\widehat{%
\rho }_{2}-1)^{2}$ are the same and given by Lemma 1(b). However,
asymp-\linebreak totically $N^{1/4}(\widehat{\rho }_{1}-1)$ and $N^{1/4}(%
\widehat{\rho }_{2}-1)$ have opposite signs. Hence a theory that can
determine the limiting distribution of the sign of $N^{1/4}(\widehat{\rho }%
-1)$ cannot exist in this case.

\subsection{GMM-based inference under second-order identification}

Dovonon, Hall and Kleibergen (2020) studied and compared the local power
properties of various test-statistics for conducting inference in moment
conditions models that locally identify the parameters only to second order.
The tests considered include tests for $H_{0}:\phi _{0}=a,$ where $a$ is a
known vector, such as the conventional Wald and LM tests, the Generalized
Anderson-Rubin (GAR) test (Anderson and Rubin, 1949; Staiger and Stock,
1997; Stock and Wright, 2000), the KLM test (Kleibergen 2002; 2005) and the
GMM extension of Moreira's (2003) Conditional LR test, also known as the
GMM-M test (Kleibergen, 2005), and tests for $H_{0}:m(\phi _{0})=0,$ such as
the identification-robust\linebreak J test of Kleibergen (2005) and the GAR
test. Under the null hypothesis the conventional LM and Wald test-statistics
have non-standard limiting distributions, although the LM test-statistic
converges to a $\chi ^{2}$ r.v. in a special case; the distribution of the
Wald test-statistic depends on $T^{1/4}(\widehat{\phi }_{p}-\phi _{0,p})$
only through $T^{1/2}(\widehat{\phi }_{p}-\phi _{0,p})^{2}$, which has a
limiting distribution that is a mixture of a half-normal distribution and $0$%
. All the other test-statistics are robust to weak and second-order local
identification and have the same limiting distribution under the null
hypothesis as they would have under first-order local identification. Apart
from the Wald test, all the tests can also be used when the Jacobian is rank
deficient by more than one. Dovonon, Hall and Kleibergen (2020) found that
in a particular panel AR(1) model, the Wald test of the unit root hypothesis
has better power than the GAR, KLM, LM and GMM-M tests.

Kruiniger (2018b) discusses a Quasi LM\ test for $H_{0}:\phi _{0}=a,$ when $%
\phi _{0,p}$ is possibly only second-order locally identified. Specifically,
Kruiniger's (2018b) Quasi LM\ test-statistic generalizes the LM
test-statistic $W_{n}^{(2)}(\theta )$ in Bottai (2003), who studied the
asymptotic behaviour of several tests and confidence regions in identifiable
one-dimensional parametric models with a smooth likelihood function and
Fisher information equal to zero at some point in the parameter space, in
two ways, namely by allowing for several parameters in the model and by
relaxing Bottai's ML setup to a Quasi ML setup. Under $H_{0}$ both LM\
test-statistics have a $\chi ^{2}$-distribution, also when one of the
parameters that appears in the null hypothesis is only second-order locally
identified. In the latter case, the score that corresponds to that parameter
and appears in the LM\ test-statistic under first-order local identification
will be replaced by its first-derivative. Kruiniger (2018b) shows that his
Quasi LM\ test and the confidence region that is based on inverting his
test-statistic have correct asymptotic size in a uniform sense.

Finally, Lee and Liao (2018) pointed out that when (a part of) $\phi _{0}$
is only second-order locally identified by the original set of moment
conditions, then Jacobian-based moment conditions can be used to obtain GMM
estimators and overidentification-test-statistics with standard asymptotic
properties. They then noted that the asymptotic normal distributions of such
GMM estimators can be used to conduct standard inference on $\phi _{0}$.
However, their tests and confidence intervals are only valid when (a part
of) $\phi _{0}$ is only second-order locally identified by the original set
of moment conditions and hence they obviously do not have correct asymptotic
size in a uniform sense.

\subsection{Monte Carlo results}

Using simulations, Kruiniger (2018a) and DH studied the finite sample
properties of specific GMM estimators under second-order identification in
the cases of exact and overidentification, respectively. Both papers found
that the GMM\ estimator of $\phi _{0,p}$\thinspace is biased. These findings
are related to the asymmetry of the limiting distributions of $T^{1/4}(%
\widehat{\phi }_{p}-\phi _{0,p})$ in these cases, which are given in Theorem
1 above.\footnote{%
Recall that in the case of exact identification Theorem 1 of DH does not
provide a limiting distribution of $T^{1/4}(\widehat{\phi }_{p}-\phi _{0,p})$
because $\Pr (%
\mathbb{R}
_{1}=0)>0$ in this case.} DH also studied the coverage rates of two
symmetric confidence intervals (CI's) for $\phi _{0,p}$ in a model with $p=1$
that are based on a GMM\ estimator for $\phi _{0,p}$ that exploits $q>p$
moment conditions and use analytic and simulated quantiles, respectively, of
the limiting distribution given in Lemma 1(b) above. Their Monte Carlo
evidence seems to suggest that the coverage rates of these CI's converge to
the nominal level except when $\phi _{0,p}$ is close to the value at which $%
\phi _{0,p}$ is second-order identified. However, DH did not claim that
these CI's have correct asymptotic size in a uniform sense. \pagebreak

\subsection{Indirect Inference}

Dovonon and Hall (2018) also considered the limiting distribution of an
Indirect Inference (II) estimator under second-order local identification.
Specifically, DH focused on an II estimator for the parameter vector $\theta
_{0}\in \Omega \subset \mathcal{%
\mathbb{R}
}^{p}$ which is defined by the following set-up: the auxiliary model
consists of a set of $q$ population moment conditions indexed by a vector of
auxiliary parameters $h\in \mathcal{H}\subset \mathcal{%
\mathbb{R}
}^{l}$ and the target function for the II estimation is a GMM estimator of
the auxiliary parameter vector. Within this framework, there are two types
of identification conditions: one set involving the binding function, and
the other involving the auxiliary parameters. The standard first-order
asymptotic theory is premised on the assumption that the binding function
satisfies global and first-order local identification conditions and the
auxiliary parameters are globally and first-order locally identified within
the auxiliary model. DH presents the limiting distribution of the II
estimator under the following two scenarios: (i) the binding function
satisfies the global and first-order local identification conditions and the
auxiliary parameters are globally identified but only locally identified at
second order; (ii) the binding function satisfies the global identification
condition but only satisfies the local identification condition at second
order, and the auxiliary parameters are globally and first-order locally
identified.

Unsurprisingly, the limiting distribution theories of DH for II estimation
under these two different scenarios with second-order identification, i.e.,
their Theorems 2 and 3(b) are incomplete for similar reasons as their
limiting distribution theory for GMM\ estimation is: their limiting
distributions for scenario (i) and scenario (ii) are only valid in the case
of overidentification, that is, when $q>l$ and when $l>p,$ respectively, or
in terms of their conditions, when $\Pr (\mathcal{R}_{1}^{(a)}=0)=0$ and $%
\Pr (\mathcal{R}_{1}^{(b)}(s)=0)=0$, respectively, where $\mathcal{R}%
_{1}^{(a)}$ and $\mathcal{R}_{1}^{(b)}(s)$ are defined similarly as $%
\mathcal{R}_{1}$, see DH. Our results and derivations for GMM estimation
under second-order local identification can be used to complete both
theories and in particular can be used to straightforwardly derive the
limiting distributions of the II estimators in the case of exact
identification under either scenario. The representations of these
distributions are obtained by replacing $\mathcal{R}_{1}^{(a)}$ and $%
\mathcal{R}_{1}^{(b)}(s)$ (implicit) in Theorems 2 and 3(b) in DH by $%
\mathcal{R}_{2}^{(a)}$ and $\mathcal{R}_{2}^{(b)}(s)$, respectively, which
are defined similarly as $\mathcal{R}_{2}$ above just like $\mathcal{R}%
_{1}^{(a)}$ and $\mathcal{R}_{1}^{(b)}(s)$ are defined similarly as $%
\mathcal{R}_{1}$.\vspace{-0.12in}

\section{Concluding remarks}

The limiting distribution theory of DH (2018) for GMM\ estimators under
second-order\ local identification depends on a non-transparent condition,
namely that $\Pr (\mathcal{R}_{1}=0)=0$ where $\mathcal{R}_{1}$ is defined
in (\ref{f7}). We have shown that this condition is only satisfied when $%
\phi $ is overidentified and derived the limiting distribution of $T^{1/4}(%
\widehat{\phi }_{p}-\phi _{0,p})$ for the case where $\phi $ is exactly
identified. This distribution is different from that of $T^{1/4}(\widehat{%
\phi }_{p}-\phi _{0,p})$ given in DH\ for the case where $\phi $ is
overidentified. In particular, the limiting distributions of the sign of $%
T^{1/4}(\widehat{\phi }_{p}-\phi _{0,p})$ for the cases of exact and
overidentification, respectively, are different and are obtained by using
expansions of the GMM objective function of different orders. We have also
pointed out that the limiting distribution theories of DH for Indirect
Inference (II) estimation under two different scenarios with second-order
identification where the target function is a GMM\ estimator of the
auxiliary parameter vector, are incomplete for similar reasons and we have
discussed how they can be completed.

The asymptotic theory for GMM\ estimators that has been discussed in this
paper can be generalized in two directions: (i) one can consider cases where
the (expected) Jacobian matrix has a rank deficiency that is higher than
one, and/or (ii) local identification of an order that is higher than two.\
In the case of second-order identification where the Jacobian has a rank
deficiency of $rd$ (with $rd\in 
\mathbb{N}
\backslash \{0,1\}$), one can reparametrize the model in such a way that the
last $rd$ columns of the Jacobian are zero, and we expect that the
asymptotic theory is similar to the theory for the case of a rank deficiency
of one apart from the fact that in the current case there are now $rd$ GMM\
estimators that converge at rate $T^{-1/4}$. In the case of local
identification of order $s$ (with $s>1$), we expect that the GMM\
estimator(s) of the higher-order identified parameter(s) converge(s) at rate 
$T^{-1/(2s)}.$ Like Rotnitzky et al. (2000), one also needs to distinguish
between cases where $s$ is even and cases where $s$ is odd: when $s$ is
even, we expect that the limiting distribution(s) of the GMM\ estimator(s)
of the higher-order identified parameter(s) is/are a mixture of a spike at
the true value and a non-standard distribution, while when $s$ is odd, we
expect that their limiting distribution(s) is/are equal to the distribution
of the $s-th$ root of a normal random variable. Furthermore, when $s$ is
even, the GMM estimators of the remaining parameters converge at the usual
rate $T^{-1/2}$ and have a limiting distribution that is a mixture of two
normal distributions, whereas when $s$ is odd, they converge at the usual
rate $T^{-1/2}$ and have a normal limiting distribution.\vspace{-0.12in}%
\pagebreak

\section{Appendix. Proofs}

\textbf{Proof of Lemma 1.}

(a) For a proof for this part we refer to the proof of part (a) of Theorem 1
in DH.

(b) Here we provide an alternative to the proof of DH. Our proof is
self-contained, unlike their proof, and it is also more straightforward than
their proof.

Using $\frac{\partial m_{T}}{\partial \phi _{p}}(\phi _{0})=O_{p}(T^{-1/2})$
and $(\widehat{\phi }_{p}-\phi _{0,p})=o_{p}(1)$ (from Proposition 1), DH
show in the proof for part (a) of their Theorem 1 that%
\begin{equation}
m_{T}(\widehat{\phi })=m_{T}(\phi _{0})+\frac{\partial m_{T}}{\partial \phi
_{1}^{\prime }}(\overline{\phi }_{1},\widehat{\phi }_{p})(\widehat{\phi }%
_{1}-\phi _{0,1})+\frac{1}{2}\frac{\partial ^{2}m_{T}}{\partial \phi _{p}^{2}%
}(\phi _{0,1},\overline{\phi }_{p})(\widehat{\phi }_{p}-\phi
_{0,p})^{2}+o_{p}(T^{-1/2}).  \label{f8}
\end{equation}%
where $\overline{\phi }_{1}\in (\phi _{0,1},\widehat{\phi }_{1})$ and may
differ from row to row, and where $\overline{\phi }_{p}\in (\phi _{0,p},%
\widehat{\phi }_{p})$ and may differ from row to row.

From (a) and (\ref{f8}), we have%
\begin{equation*}
m_{T}(\widehat{\phi })=m_{T}(\phi _{0})+D(\widehat{\phi }_{1}-\phi _{0,1})+%
\frac{1}{2}G(\widehat{\phi }_{p}-\phi _{0,p})^{2}+o_{p}(T^{-1/2}).
\end{equation*}%
The first-order condition for an interior solution is given by:%
\begin{equation*}
\frac{\partial m_{T}^{\prime }}{\partial \phi }(\widehat{\phi })W_{T}m_{T}(%
\widehat{\phi })=0.
\end{equation*}%
In the direction of $\phi _{1}$, this amounts to%
\begin{equation*}
(D^{\prime }+o_{p}(1))W(\sqrt{T}m_{T}(\phi _{0})+D\sqrt{T}(\widehat{\phi }%
_{1}-\phi _{0,1})+\frac{1}{2}G\sqrt{T}(\widehat{\phi }_{p}-\phi
_{0,p})^{2}+o_{p}(1))=0.
\end{equation*}%
This gives:%
\begin{equation}
\sqrt{T}(\widehat{\phi }_{1}-\phi _{0,1})=-(D^{\prime }WD)^{-1}D^{\prime }W(%
\sqrt{T}m_{T}(\phi _{0})+\frac{1}{2}G\sqrt{T}(\widehat{\phi }_{p}-\phi
_{0,p})^{2})+o_{p}(1).  \label{g3}
\end{equation}%
A Taylor expansion of $m_{T}^{\prime }(\widehat{\phi })W_{T}m_{T}(\widehat{%
\phi })$ around $\phi _{0}$ up to second-order gives:%
\begin{equation}
Q_{T}(\widehat{\phi })=m_{T}^{\prime }(\widehat{\phi })W_{T}m_{T}(\widehat{%
\phi })=m_{T}^{\prime }(\widehat{\phi })Wm_{T}(\widehat{\phi }%
)+o_{p}(T^{-1})=\underline{K}_{T}(\phi _{0})+o_{p}(T^{-1})  \label{g2}
\end{equation}%
where%
\begin{eqnarray}
\underline{K}_{T}(\phi _{0}) &=&m_{T}^{\prime }(\phi _{0})Wm_{T}(\phi
_{0})+2m_{T}^{\prime }(\phi _{0})W(D(\widehat{\phi }_{1}-\phi _{0,1})+\frac{1%
}{2}G(\widehat{\phi }_{p}-\phi _{0,p})^{2})  \notag \\
&&+(D(\widehat{\phi }_{1}-\phi _{0,1})+\frac{1}{2}G(\widehat{\phi }_{p}-\phi
_{0,p})^{2})^{\prime }W(D(\widehat{\phi }_{1}-\phi _{0,1})+\frac{1}{2}G(%
\widehat{\phi }_{p}-\phi _{0,p})^{2}).\qquad  \label{g4}
\end{eqnarray}%
Defining $Z_{0T}=m_{T}(\phi _{0})$ and replacing $(\widehat{\phi }_{1}-\phi
_{0,1})$ in (\ref{g4}) by its expression from (\ref{g3}), the leading $%
O_{p}(T^{-1})$ term in the expansion of $m_{T}^{\prime }(\widehat{\phi }%
)W_{T}m_{T}(\widehat{\phi })$ is obtained as $\underline{K}_{T}(\phi _{0,p})$
with 
\begin{equation}
\underline{K}_{T}(\phi _{0,p})=Z_{0T}^{\prime
}W^{1/2}M_{d}W^{1/2}Z_{0T}+Z_{0T}^{\prime }W^{1/2}M_{d}W^{1/2}G(\widehat{%
\phi }_{p}-\phi _{0,p})^{2}+\frac{1}{4}G^{\prime }W^{1/2}M_{d}W^{1/2}G(%
\widehat{\phi }_{p}-\phi _{0,p})^{4}.  \label{g5}
\end{equation}%
Let $Z_{T}=Z_{0T}^{\prime }W^{1/2}M_{d}W^{1/2}G.$ If $Z_{T}<0,$ then $%
m_{T}^{\prime }(\widehat{\phi })W_{T}m_{T}(\widehat{\phi })$ is minimized at%
\begin{equation}
(\widehat{\phi }_{p}-\phi _{0,p})^{2}=-2\frac{Z_{0T}^{\prime
}W^{1/2}M_{d}W^{1/2}G}{G^{\prime }W^{1/2}M_{d}W^{1/2}G}+o_{p}(T^{-1/2}).
\label{g6}
\end{equation}%
If $Z_{T}\geq 0,$ then $m_{T}^{\prime }(\widehat{\phi })W_{T}m_{T}(\widehat{%
\phi })$ is minimized at $(\widehat{\phi }_{p}-\phi
_{0,p})^{2}=o_{p}(T^{-1/2}).$

Since $\sqrt{T}m_{T}(\phi _{0})$ and $\sqrt{T}(\widehat{\phi }_{p}-\phi
_{0,p})^{2}$ are $O_{p}(1)$, the pair is jointly $O_{p}(1)$ and by
Prohorov's theorem, any subsequence of them has a further subsequence that
jointly converges in distribution towards, say, $(\mathcal{Z}_{0},V)$. Hence,%
\begin{equation*}
Tm_{T}^{\prime }(\widehat{\phi })W_{T}m_{T}(\widehat{\phi })\overset{d}{%
\rightarrow }\underline{K}(\phi _{0})\equiv \mathcal{Z}_{0}^{\prime
}W^{1/2}M_{d}W^{1/2}\mathcal{Z}_{0}+\mathcal{Z}_{0}^{\prime
}W^{1/2}M_{d}W^{1/2}GV+\frac{1}{4}G^{\prime }W^{1/2}M_{d}W^{1/2}GV^{2}.
\end{equation*}%
Let $\mathcal{Z}=\mathcal{Z}_{0}^{\prime }W^{1/2}M_{d}W^{1/2}G.$ If $%
\mathcal{Z}<0,$ then $\underline{K}(\phi _{0})$ is minimized at 
\begin{equation*}
V=-2\frac{\mathcal{Z}_{0}^{\prime }W^{1/2}M_{d}W^{1/2}G}{G^{\prime
}W^{1/2}M_{d}W^{1/2}G}=-2\frac{\mathcal{Z}}{\sigma _{G}}
\end{equation*}%
If $\mathcal{Z}\geq 0,$ then $\underline{K}(\phi _{0})$ is minimized at $%
V=0. $

Either way, we conclude that $\sqrt{T}(\widehat{\phi }_{p}-\phi _{0,p})^{2}%
\overset{d}{\rightarrow }V=-2\mathcal{Z}\mathbf{1}(\mathcal{Z}<0)/\sigma
_{G}.$

Finally, using (\ref{g3}) we obtain that $\sqrt{T}(\widehat{\phi }_{1}-\phi
_{0,1})\overset{d}{\rightarrow }H(\mathcal{Z}_{0}+\frac{1}{2}GV),$ where $%
H=-(D^{\prime }WD)^{-1}D^{\prime }W.$\quad ${\tiny \blacksquare }$\textbf{%
\medskip }\pagebreak

\textbf{Proof of Theorem 1.}

The limiting distribution for $T^{1/4}(\widehat{\phi }_{p}-\phi _{0,p})$
that is specified in part (c) of Theorem 1 of DH (2018) is only valid in the
case of overidentification, i.e., when $q>p,$ and DH's proof of this claim
is only valid when $q>p.$ Here we provide a more straightforward and more
complete proof of this claim for the case $q>p$, i.e., of part (a) of our
Theorem 1, and also derive the limiting distribution of $T^{1/4}(\widehat{%
\phi }_{p}-\phi _{0,p})$ in the case of exact identification, i.e., when $%
q=p,$ which is given in part (b) of our Theorem 1.

Proof of (a): In our proof for part (b) of Lemma 1 we have derived the
limiting distribution of $\sqrt{T}(\widehat{\phi }_{p}-\phi _{0,p})^{2}.$ To
get the limiting distribution of $T^{1/4}(\widehat{\phi }_{p}-\phi _{0,p}),$
it remains to characterize its sign when $\mathcal{Z}<0$.

The powers of $(\widehat{\phi }_{p}-\phi _{0,p})$ in an expansion of $%
m_{T}^{\prime }(\widehat{\phi })W_{T}m_{T}(\widehat{\phi })$ around $\phi
_{0}$ up to $O_{p}(T^{-1})$ are even, cf. $\underline{K}_{T}(\phi _{0,p})$
in (\ref{g5}), and therefore we cannot characterize the sign of $T^{1/4}(%
\widehat{\phi }_{p}-\phi _{0,p})$ by using such an expansion. However,
following the approach of Rotnitzky et al. (2000) for the ML estimator, we
can do this when $q>p$ by expanding $m_{T}^{\prime }(\widehat{\phi }%
)W_{T}m_{T}(\widehat{\phi })$ up to $o_{p}(T^{-5/4})$. Specifically, the $%
O_{p}(T^{-5/4})$ terms in a Taylor expansion of $m_{T}^{\prime }(\widehat{%
\phi })W_{T}m_{T}(\widehat{\phi })$ up to $o_{p}(T^{-5/4})$ will provide the
sign of $T^{1/4}(\widehat{\phi }_{p}-\phi _{0,p})$ in case $q>p$ as we will
now show.\vspace*{-0.05in}%
\begin{equation*}
Q_{T}(\widehat{\phi })=m_{T}^{\prime }(\widehat{\phi })W_{T}m_{T}(\widehat{%
\phi })=m_{T}^{\prime }(\widehat{\phi })Wm_{T}(\widehat{\phi }%
)+o_{p}(T^{-5/4})=\underline{K}_{T}(\phi _{0,p})+\underline{R}_{1T}(\phi
_{0})+o_{p}(T^{-5/4})
\end{equation*}%
\vspace*{-0.05in}with $\underline{R}_{1T}(\phi _{0})=(\widehat{\phi }%
_{p}-\phi _{0,p})\times 2R_{1T}$ where\vspace*{-0.05in}%
\begin{eqnarray}
R_{1T} &=&(m_{T}(\phi _{0})+\frac{\partial m_{T}}{\partial \phi _{1}^{\prime
}}(\phi _{0})(\widehat{\phi }_{1}-\phi _{0,1}))^{\prime }W(\frac{\partial
m_{T}}{\partial \phi _{p}}(\phi _{0})+\frac{\partial ^{2}m_{T}}{\partial
\phi _{p}\partial \phi _{1}^{\prime }}(\phi _{0})(\widehat{\phi }_{1}-\phi
_{0,1}))+  \notag \\
&&\frac{1}{3!}(m_{T}(\phi _{0})+\frac{\partial m_{T}}{\partial \phi
_{1}^{\prime }}(\phi _{0})(\widehat{\phi }_{1}-\phi _{0,1}))^{\prime }W(%
\frac{\partial ^{3}m_{T}}{\partial \phi _{p}^{3}}(\phi _{0}))(\widehat{\phi }%
_{p}-\phi _{0,p})^{2}+  \notag \\
&&\frac{1}{2!}(\frac{\partial m_{T}}{\partial \phi _{p}}(\phi _{0})+\frac{%
\partial ^{2}m_{T}}{\partial \phi _{p}\partial \phi _{1}^{\prime }}(\phi
_{0})(\widehat{\phi }_{1}-\phi _{0,1}))^{\prime }W(\frac{\partial ^{2}m_{T}}{%
\partial \phi _{p}^{2}}(\phi _{0}))(\widehat{\phi }_{p}-\phi _{0,p})^{2}+ 
\notag \\
&&\frac{1}{3!2!}(\frac{\partial ^{3}m_{T}}{\partial \phi _{p}^{3}}(\phi
_{0}))^{\prime }W(\frac{\partial ^{2}m_{T}}{\partial \phi _{p}^{2}}(\phi
_{0}))(\widehat{\phi }_{p}-\phi _{0,p})^{4}.  \label{g7}
\end{eqnarray}%
\vspace*{-0.05in}Defining $Z_{1T}=\frac{\partial m_{T}}{\partial \phi _{p}}%
(\phi _{0})$ and replacing $(\widehat{\phi }_{1}-\phi _{0,1})$ in (\ref{g7})
by its expression from (\ref{g3}), we obtain\vspace*{-0.05in}%
\begin{eqnarray}
R_{1T} &=&(Z_{0T}+DH(Z_{0T}+\frac{1}{2}G(\widehat{\phi }_{p}-\phi
_{0,p})^{2}))^{\prime }W(Z_{1T}+G_{1p}H(Z_{0T}+\frac{1}{2}G(\widehat{\phi }%
_{p}-\phi _{0,p})^{2}))+  \notag \\
&&\frac{1}{3!}(Z_{0T}+DH(Z_{0T}+\frac{1}{2}G(\widehat{\phi }_{p}-\phi
_{0,p})^{2}))^{\prime }WL(\widehat{\phi }_{p}-\phi _{0,p})^{2}+  \notag
\end{eqnarray}%
\begin{eqnarray}
&&\frac{1}{2!}(Z_{1T}+G_{1p}H(Z_{0T}+\frac{1}{2}G(\widehat{\phi }_{p}-\phi
_{0,p})^{2}))^{\prime }WG(\widehat{\phi }_{p}-\phi _{0,p})^{2}+  \notag \\
&&\frac{1}{3!2!}L^{\prime }WG(\widehat{\phi }_{p}-\phi
_{0,p})^{4}+o_{p}(T^{-1})  \notag \\
&=&Z_{0T}^{\prime }W^{1/2}M_{d}W^{1/2}Z_{1T}+Z_{0T}^{\prime
}W^{1/2}M_{d}W^{1/2}G_{1p}HZ_{0T}+  \notag \\
&&(\frac{1}{3}Z_{0T}^{\prime }W^{1/2}M_{d}W^{1/2}L+Z_{1T}^{\prime
}W^{1/2}M_{d}W^{1/2}G+  \notag \\
&&G^{\prime }W^{1/2}M_{d}W^{1/2}G_{1p}HZ_{0T}+Z_{0T}^{\prime
}W^{1/2}M_{d}W^{1/2}G_{1p}HG)(\widehat{\phi }_{p}-\phi _{0,p})^{2}+  \notag
\\
&&(\frac{1}{6}G^{\prime }W^{1/2}M_{d}W^{1/2}L+\frac{1}{2}G^{\prime
}W^{1/2}M_{d}W^{1/2}G_{1p}HG)(\widehat{\phi }_{p}-\phi
_{0,p})^{4}+o_{p}(T^{-1}).  \label{g8}
\end{eqnarray}%
At the minimum of $Q_{T}(\phi )$, we expect $\underline{R}_{1T}(\phi _{0})$
to be negative, i.e., $(\widehat{\phi }_{p}-\phi _{0,p})$ and $R_{1T}$ have
opposite sign. Hence,\vspace*{-0.07in}%
\begin{equation*}
T^{1/4}(\widehat{\phi }_{p}-\phi _{0,p})=(-1)^{B_{1T}}T^{1/4}|\widehat{\phi }%
_{p}-\phi _{0,p}|,
\end{equation*}%
\vspace*{-0.07in}with $B_{1T}=\mathbf{1(}TR_{1T}\geq 0\mathbf{)}$. After
replacing $(\widehat{\phi }_{p}-\phi _{0,p})^{2}$ in (\ref{g8}) by its
expression from (\ref{g6}) and scaling (\ref{g8}) by $T$, we can see, using
the continuous mapping theorem, that $TR_{1T}$ converges in distribution
towards $\mathcal{R}_{1}$:\vspace*{-0.07in}%
\begin{eqnarray}
\mathcal{R}_{1} &=&(\mathcal{Z}_{0}^{\prime }W^{1/2}P_{g}W^{1/2}\mathcal{Z}%
_{0}G^{\prime }-G^{\prime }W^{1/2}P_{g}W^{1/2}\mathcal{Z}_{0}\mathcal{Z}%
_{0}^{\prime })W^{1/2}M_{d}W^{1/2}(\frac{1}{3}L+G_{1p}HG)/\sigma _{G}+ 
\notag \\
&&\mathcal{Z}_{0}^{\prime }W^{1/2}M_{dg}W^{1/2}(\mathcal{Z}_{1}+G_{1p}H%
\mathcal{Z}_{0}).  \label{g9}
\end{eqnarray}%
\vspace*{-0.07in}We actually have that $(\sqrt{T}Z_{0T},$ $\sqrt{T}Z_{1T},$ $%
TR_{1T})$ converges in distribution towards\linebreak $(\mathcal{Z}_{0},$ $%
\mathcal{Z}_{1},$ $\mathcal{R}_{1})$. According to our Lemma 2, when $q>p,$ $%
\mathcal{R}_{1}$ does not have an atom of probability at $0.$

Applying a version of Lemma 1 of DH (2018), we have $(\sqrt{T}Z_{0T},$ $%
\sqrt{T}Z_{1T},$ $(-1)^{B_{1T}})\overset{d}{\rightarrow }(\mathcal{Z}_{0},$ $%
\mathcal{Z}_{1},$ $(-1)^{B_{1}})$, where $B_{1}=\mathbf{1}(\mathcal{R}%
_{1}\geq 0)$. Since $(\sqrt{T}(\widehat{\phi }_{1}-\phi _{0,1}),$ $T^{1/4}|%
\widehat{\phi }_{p}-\phi _{0,p}|,$ $(-1)^{B_{1T}})=O_{p}(1)$, any
subsequence of the left hand side has a further subsequence that converges
in distribution. Using part (b) of Lemma 1, such a subsequence obeys $(\sqrt{%
T}(\widehat{\phi }_{1}-\phi _{0,1}),$ $T^{1/4}|\widehat{\phi }_{p}-\phi
_{0,p}|,$ $(-1)^{B_{1T}})\overset{d}{\rightarrow }(H\mathcal{Z}_{0}+HGV/2,$ $%
\sqrt{V},$ $(-1)^{B_{1}})$. Since the limit distribution does not depend on
the subsequence, the whole sequence converges towards that limit. By the
continuous mapping theorem, we deduce that: $(\sqrt{T}(\widehat{\phi }%
_{1}-\phi _{0,1}),$ $T^{1/4}(\widehat{\phi }_{p}-\phi _{0,p}))\overset{d}{%
\rightarrow }(H\mathcal{Z}_{0}+HGV/2,$ $(-1)^{B_{1}}\sqrt{V})$.

Proof of (b): Part (b) of Lemma 1 gives the limiting distribution of $\sqrt{T%
}(\widehat{\phi }_{p}-\phi _{0,p})^{2}.$ To get the limiting distribution of 
$T^{1/4}(\widehat{\phi }_{p}-\phi _{0,p}),$ it remains to characterize its
sign when $\mathcal{Z}<0$. According to our Lemma 2, when $q=p,$ $\Pr (%
\mathcal{R}_{1}=0|\mathcal{Z}<0\mathbf{)}=1.$ Hence when $q=p,$ $%
(-1)^{B_{1}} $ with $B_{1}=\mathbf{1}(\mathcal{R}_{1}\geq 0)$ will not
correctly describe the behaviour of the sign of $T^{1/4}(\widehat{\phi }%
_{p}-\phi _{0,p})$ in its limiting distribution. However, when $q=p,$ we can
characterize the sign of $T^{1/4}(\widehat{\phi }_{p}-\phi _{0,p})$ by
expanding $m_{T}^{\prime }(\widehat{\phi })W_{T}m_{T}(\widehat{\phi })$
around $\phi _{0}$ up to $o_{p}(T^{-7/4})$. Specifically, the $%
O_{p}(T^{-7/4})$ terms in a Taylor expansion of $m_{T}^{\prime }(\widehat{%
\phi })W_{T}m_{T}(\widehat{\phi })$ up to $o_{p}(T^{-7/4})$ will provide the
sign of $T^{1/4}(\widehat{\phi }_{p}-\phi _{0,p})$ in case $q=p$ as we will
now show.

First we note that the powers of $(\widehat{\phi }_{p}-\phi _{0,p})$ in the $%
O_{p}(T^{-6/4})$ terms in an expansion of $m_{T}^{\prime }(\widehat{\phi }%
)W_{T}m_{T}(\widehat{\phi })$ up to $o_{p}(T^{-7/4})$ are all even and
therefore these terms do not provide information on the sign of $T^{1/4}(%
\widehat{\phi }_{p}-\phi _{0,p})$. Next, we consider the $O_{p}(T^{-7/4})$
terms in the Taylor expansion of $m_{T}^{\prime }(\widehat{\phi })W_{T}m_{T}(%
\widehat{\phi })$ around $\phi _{0}:$

$\underline{R}_{5,T}(\phi _{0})=(\widehat{\phi }_{p}-\phi _{0,p})\times
2R_{5,T}$ where%
\begin{eqnarray}
R_{5,T} &=&\frac{1}{2!}(\kappa _{1,T}\ldots \kappa _{q,T})W(\frac{\partial
m_{T}}{\partial \phi _{p}}(\phi _{0})+\frac{\partial ^{2}m_{T}}{\partial
\phi _{p}\partial \phi _{1}^{\prime }}(\phi _{0})(\widehat{\phi }_{1}-\phi
_{0,1}))+  \notag \\
&&\frac{1}{2!3!}(\kappa _{1,T}\ldots \kappa _{q,T})W(\frac{\partial ^{3}m_{T}%
}{\partial \phi _{p}^{3}}(\phi _{0}))(\widehat{\phi }_{p}-\phi _{0,p})^{2}+ 
\notag \\
&&\frac{1}{3!}(m_{T}(\phi _{0})+\frac{\partial m_{T}}{\partial \phi
_{1}^{\prime }}(\phi _{0})(\widehat{\phi }_{1}-\phi _{0,1}))^{\prime }W(%
\frac{\partial ^{4}m_{T}}{\partial \phi _{p}^{3}\partial \phi _{1}^{\prime }}%
(\phi _{0})(\widehat{\phi }_{1}-\phi _{0,1}))(\widehat{\phi }_{p}-\phi
_{0,p})^{2}+  \notag \\
&&\frac{1}{2!}(\frac{\partial m_{T}}{\partial \phi _{p}}(\phi _{0})+\frac{%
\partial ^{2}m_{T}}{\partial \phi _{p}\partial \phi _{1}^{\prime }}(\phi
_{0})(\widehat{\phi }_{1}-\phi _{0,1}))^{\prime }W(\frac{\partial ^{3}m_{T}}{%
\partial \phi _{p}^{2}\partial \phi _{1}^{\prime }}(\phi _{0})(\widehat{\phi 
}_{1}-\phi _{0,1}))(\widehat{\phi }_{p}-\phi _{0,p})^{2}+  \notag \\
&&\frac{1}{2!3!}(\frac{\partial ^{2}m_{T}}{\partial \phi _{p}^{2}}(\phi
_{0}))^{\prime }W(\frac{\partial ^{4}m_{T}}{\partial \phi _{p}^{3}\partial
\phi _{1}^{\prime }}(\phi _{0})(\widehat{\phi }_{1}-\phi _{0,1}))(\widehat{%
\phi }_{p}-\phi _{0,p})^{4}+  \notag \\
&&\frac{1}{3!2!}(\frac{\partial ^{3}m_{T}}{\partial \phi _{p}^{3}}(\phi
_{0}))^{\prime }W(\frac{\partial ^{3}m_{T}}{\partial \phi _{p}^{2}\partial
\phi _{1}^{\prime }}(\phi _{0})(\widehat{\phi }_{1}-\phi _{0,1}))(\widehat{%
\phi }_{p}-\phi _{0,p})^{4}+  \notag \\
&&\frac{1}{5!}(m_{T}(\phi _{0})+\frac{\partial m_{T}}{\partial \phi
_{1}^{\prime }}(\phi _{0})(\widehat{\phi }_{1}-\phi _{0,1}))^{\prime }W(%
\frac{\partial ^{5}m_{T}}{\partial \phi _{p}^{5}}(\phi _{0}))(\widehat{\phi }%
_{p}-\phi _{0,p})^{4}+  \notag \\
&&\frac{1}{4!}(\frac{\partial m_{T}}{\partial \phi _{p}}(\phi _{0})+\frac{%
\partial ^{2}m_{T}}{\partial \phi _{p}\partial \phi _{1}^{\prime }}(\phi
_{0})(\widehat{\phi }_{1}-\phi _{0,1}))^{\prime }W(\frac{\partial ^{4}m_{T}}{%
\partial \phi _{p}^{4}}(\phi _{0}))(\widehat{\phi }_{p}-\phi _{0,p})^{4}+ 
\notag \\
&&\frac{1}{2!5!}(\frac{\partial ^{2}m_{T}}{\partial \phi _{p}^{2}}(\phi
_{0}))^{\prime }W(\frac{\partial ^{5}m_{T}}{\partial \phi _{p}^{5}}(\phi
_{0}))(\widehat{\phi }_{p}-\phi _{0,p})^{6}+  \notag \\
&&\frac{1}{3!4!}(\frac{\partial ^{3}m_{T}}{\partial \phi _{p}^{3}}(\phi
_{0}))^{\prime }W(\frac{\partial ^{4}m_{T}}{\partial \phi _{p}^{4}}(\phi
_{0}))(\widehat{\phi }_{p}-\phi _{0,p})^{6}.  \label{h1}
\end{eqnarray}%
with $\kappa _{k,T}=(\widehat{\phi }_{1}-\phi _{0,1})^{\prime }K_{k,T}(%
\widehat{\phi }_{1}-\phi _{0,1})$ and $K_{k,T}=\frac{\partial ^{2}m_{k,T}}{%
\partial \phi _{1}\partial \phi _{1}^{\prime }}(\phi _{0})$ for $%
k=1,2,\ldots ,q,$ where $m_{k,T}(\phi )$ is the kth element of $m_{T}(\phi
). $

Recalling that when $q=p,$ then $m(\widehat{\phi })=0$ and hence $m_{T}(\phi
_{0})+\frac{\partial m_{T}}{\partial \phi _{1}^{\prime }}(\phi _{0})(%
\widehat{\phi }_{1}-\phi _{0,1})+\frac{1}{2}\frac{\partial ^{2}m_{T}}{%
\partial \phi _{p}^{2}}(\phi _{0})(\widehat{\phi }_{p}-\phi
_{0,p})^{2}=o_{p}(T^{-1/2})$ (cf. (\ref{f8})), and replacing $(\widehat{\phi 
}_{1}-\phi _{0,1})$ in (\ref{h1}) by its expression from (\ref{g3}), we
obtain%
\begin{eqnarray}
R_{5,T} &=&\frac{1}{2!}(\lambda _{1,T}\ldots \lambda
_{q,T})W(Z_{1T}+G_{1p}H(Z_{0T}+\frac{1}{2}G(\widehat{\phi }_{p}-\phi
_{0,p})^{2}))+  \notag \\
&&\frac{1}{2!3!}(\lambda _{1,T}\ldots \lambda _{q,T})WL(\widehat{\phi }%
_{p}-\phi _{0,p})^{2}+  \notag \\
&&\frac{1}{2!}(Z_{1T}+G_{1p}H(Z_{0T}+\frac{1}{2}G(\widehat{\phi }_{p}-\phi
_{0,p})^{2}))^{\prime }W(G_{1pp}H(Z_{0T}+\frac{1}{2}G(\widehat{\phi }%
_{p}-\phi _{0,p})^{2}))(\widehat{\phi }_{p}-\phi _{0,p})^{2}+  \notag \\
&&\frac{1}{3!2!}L^{\prime }W(G_{1pp}H(Z_{0T}+\frac{1}{2}G(\widehat{\phi }%
_{p}-\phi _{0,p})^{2}))(\widehat{\phi }_{p}-\phi _{0,p})^{4}+  \notag \\
&&\frac{1}{4!}(Z_{1T}+G_{1p}H(Z_{0T}+\frac{1}{2}G(\widehat{\phi }_{p}-\phi
_{0,p})^{2}))^{\prime }WF(\widehat{\phi }_{p}-\phi _{0,p})^{4}+  \notag \\
&&\frac{1}{3!4!}L^{\prime }WF(\widehat{\phi }_{p}-\phi
_{0,p})^{6}+o_{p}(T^{-6/4}).  \label{h2}
\end{eqnarray}%
with $\lambda _{k,T}=\lambda _{1,k,T}+\lambda _{2,k,T}(\widehat{\phi }%
_{p}-\phi _{0,p})^{2}+\lambda _{3,k}(\widehat{\phi }_{p}-\phi _{0,p})^{4},$
where $\lambda _{1,k,T}=Z_{0T}^{\prime }H^{\prime }K_{k}HZ_{0T},$ $\lambda
_{2,k,T}=G^{\prime }H^{\prime }K_{k}HZ_{0T}$ and $\lambda _{3,k}=\frac{1}{4}%
G^{\prime }H^{\prime }K_{k}HG$ for $k=1,2,\ldots ,q.$ Let $\tilde{\lambda}%
_{i,T}=(\lambda _{i,1,T}\ldots \lambda _{i,q,T})^{\prime }$ for $i=1,2$ and
recall that $\tilde{\lambda}_{3}=(\lambda _{3,1}\ldots \lambda
_{3,q})^{\prime }.$ Combining powers in (\ref{h2}) we get%
\begin{eqnarray}
R_{5,T} &=&\frac{1}{2!}\tilde{\lambda}_{1,T}^{\prime
}W(Z_{1T}+G_{1p}HZ_{0T})+  \notag \\
&&\frac{1}{2!}\tilde{\lambda}_{1,T}^{\prime }W(\frac{1}{2}G_{1p}HG+\frac{1}{%
3!}L)(\widehat{\phi }_{p}-\phi _{0,p})^{2}+  \notag \\
&&\frac{1}{2!}\tilde{\lambda}_{2,T}^{\prime }W(Z_{1T}+G_{1p}HZ_{0T})(%
\widehat{\phi }_{p}-\phi _{0,p})^{2}+  \notag \\
&&\frac{1}{2!}\tilde{\lambda}_{2,T}^{\prime }W(\frac{1}{2}G_{1p}HG+\frac{1}{%
3!}L)(\widehat{\phi }_{p}-\phi _{0,p})^{4}+  \notag \\
&&\frac{1}{2!}\tilde{\lambda}_{3}^{\prime }W(Z_{1T}+G_{1p}HZ_{0T})(\widehat{%
\phi }_{p}-\phi _{0,p})^{4}+  \notag \\
&&\frac{1}{2!}\tilde{\lambda}_{3}^{\prime }W(\frac{1}{2}G_{1p}HG+\frac{1}{3!}%
L)(\widehat{\phi }_{p}-\phi _{0,p})^{6}+  \notag \\
&&\frac{1}{2!}(Z_{1T}+G_{1p}HZ_{0T})^{\prime }W(G_{1pp}HZ_{0T})(\widehat{%
\phi }_{p}-\phi _{0,p})^{2}+  \notag \\
&&\frac{1}{2!}(Z_{1T}+G_{1p}HZ_{0T})^{\prime }W(\frac{1}{2}G_{1pp}HG)(%
\widehat{\phi }_{p}-\phi _{0,p})^{4}+  \notag \\
&&\frac{1}{2!}(\frac{1}{2}G_{1p}HG)^{\prime }W(G_{1pp}HZ_{0T})(\widehat{\phi 
}_{p}-\phi _{0,p})^{4}+  \notag \\
&&\frac{1}{2!}(\frac{1}{2}G_{1p}HG)^{\prime }W(\frac{1}{2}G_{1pp}HG)(%
\widehat{\phi }_{p}-\phi _{0,p})^{6}+  \notag \\
&&\frac{1}{3!2!}L^{\prime }W(G_{1pp}HZ_{0T})(\widehat{\phi }_{p}-\phi
_{0,p})^{4}+\frac{1}{3!2!}L^{\prime }W(\frac{1}{2}G_{1pp}HG)(\widehat{\phi }%
_{p}-\phi _{0,p})^{6}+  \notag \\
&&\frac{1}{4!}(Z_{1T}+G_{1p}HZ_{0T})^{\prime }WF(\widehat{\phi }_{p}-\phi
_{0,p})^{4}+\frac{1}{4!}(\frac{1}{2}G_{1p}HG)^{\prime }WF(\widehat{\phi }%
_{p}-\phi _{0,p})^{6}+  \notag \\
&&\frac{1}{3!4!}L^{\prime }WF(\widehat{\phi }_{p}-\phi
_{0,p})^{6}+o_{p}(T^{-6/4}).  \label{h3}
\end{eqnarray}

However, apart from $R_{5,T}$, the `$O_{p}(T^{-5/4})$ and $O_{p}(T^{-6/4})$
terms' in the Taylor expansion of $m_{T}^{\prime }(\widehat{\phi }%
)W_{T}m_{T}(\widehat{\phi })$ around $\phi _{0}$ also contribute to the $%
O_{p}(T^{-7/4})$ terms in the Taylor expansion of $m_{T}^{\prime }(\widehat{%
\phi })W_{T}m_{T}(\widehat{\phi })$ because they contain the random vector $%
\sqrt{T}m_{T}(\phi _{0})+D\sqrt{T}(\widehat{\phi }_{1}-\phi _{0,1})+\frac{1}{%
2}G\sqrt{T}(\widehat{\phi }_{p}-\phi _{0,p})^{2}$ which equals $(\widehat{%
\phi }_{p}-\phi _{0,p})S_{1,T}+T^{-1/2}S_{2,T}+o_{p}(T^{-1/2})$ where $%
S_{1,T}=O_{p}(1)$ and $S_{2,T}=O_{p}(1)$ with $S_{1,T}$ and $\underline{S}%
_{2,T}$ known, i.e., they are determined by the application, $S_{1,T}\overset%
{d}{\rightarrow }S_{1},$ and $S_{2,T}\overset{d}{\rightarrow }S_{2}$. The `$%
O_{p}(T^{-5/4})$ term' $\underline{R}_{1T}(\phi _{0})=(\widehat{\phi }%
_{p}-\phi _{0,p})\times 2R_{1T}=2(\widehat{\phi }_{p}-\phi _{0,p})^{2}T^{-1}%
\widetilde{R}_{1T}(\sqrt{T}(\widehat{\phi }_{p}-\phi _{0,p})^{2},S_{1,T})+2(%
\widehat{\phi }_{p}-\phi _{0,p})T^{-3/2}\widetilde{R}_{1T}(\sqrt{T}(\widehat{%
\phi }_{p}-\phi _{0,p})^{2},S_{2,T})+o_{p}(T^{-7/4})$ where 
\begin{eqnarray}
\widetilde{R}_{1T}(\sqrt{T}(\widehat{\phi }_{p}-\phi _{0,p})^{2},S_{T}) &=&%
\sqrt{T}(\frac{\partial m_{T}}{\partial \phi _{p}}(\phi _{0})+\frac{\partial
^{2}m_{T}}{\partial \phi _{p}\partial \phi _{1}^{\prime }}(\phi _{0})(%
\widehat{\phi }_{1}-\phi _{0,1}))^{\prime }WS_{T}+  \notag \\
&&\frac{1}{3!}(\frac{\partial ^{3}m_{T}}{\partial \phi _{p}^{3}}(\phi
_{0}))^{\prime }WS_{T}\sqrt{T}(\widehat{\phi }_{p}-\phi _{0,p})^{2}.
\label{r1til}
\end{eqnarray}%
The `$O_{p}(T^{-6/4})$ term' is given by $\underline{R}_{4,T}(\phi
_{0})=T^{-3/2}R_{4,T}(\sqrt{T}(\widehat{\phi }_{p}-\phi
_{0,p})^{2})+T^{-3/2}\times $\linebreak $(\widehat{\phi }_{p}-\phi _{0,p})%
\widetilde{R}_{4,T}(\sqrt{T}(\widehat{\phi }_{p}-\phi
_{0,p})^{2},S_{1,T})+o_{p}(T^{-7/4})$ where%
\begin{eqnarray}
&&R_{4,T}(\sqrt{T}(\widehat{\phi }_{p}-\phi _{0,p})^{2})=T((\frac{\partial
m_{T}}{\partial \phi _{p}}(\phi _{0})+\frac{\partial ^{2}m_{T}}{\partial
\phi _{p}\partial \phi _{1}^{\prime }}(\phi _{0})(\widehat{\phi }_{1}-\phi
_{0,1}))^{\prime }W(\frac{\partial m_{T}}{\partial \phi _{p}}(\phi _{0})+%
\text{ }  \notag \\
&&\quad \frac{\partial ^{2}m_{T}}{\partial \phi _{p}\partial \phi
_{1}^{\prime }}(\phi _{0})(\widehat{\phi }_{1}-\phi _{0,1}))+\frac{1}{2!}%
(\kappa _{1,T}\ldots \kappa _{q,T})^{\prime }W\frac{\partial ^{2}m_{T}}{%
\partial \phi _{p}^{2}}(\phi _{0}))\sqrt{T}(\widehat{\phi }_{p}-\phi
_{0,p})^{2}+  \notag \\
&&\frac{2}{3!}\sqrt{T}(\frac{\partial m_{T}}{\partial \phi _{p}}(\phi _{0})+%
\frac{\partial ^{2}m_{T}}{\partial \phi _{p}\partial \phi _{1}^{\prime }}%
(\phi _{0})(\widehat{\phi }_{1}-\phi _{0,1}))^{\prime }W\frac{\partial
^{3}m_{T}}{\partial \phi _{p}^{3}}(\phi _{0})T(\widehat{\phi }_{p}-\phi
_{0,p})^{4}+  \notag \\
&&\frac{1}{3!3!}(\frac{\partial ^{3}m_{T}}{\partial \phi _{p}^{3}}(\phi
_{0}))^{\prime }W\frac{\partial ^{3}m_{T}}{\partial \phi _{p}^{3}}(\phi
_{0})T^{3/2}(\widehat{\phi }_{p}-\phi _{0,p})^{6}\quad  \label{r4}
\end{eqnarray}%
and%
\begin{eqnarray}
\widetilde{R}_{4,T}(\sqrt{T}(\widehat{\phi }_{p}-\phi _{0,p})^{2},S_{T}) &=&%
\sqrt{T}(\frac{\partial ^{3}m_{T}}{\partial \phi _{p}^{2}\partial \phi
_{1}^{\prime }}(\phi _{0})(\widehat{\phi }_{1}-\phi _{0,1}))^{\prime }WS_{T}%
\sqrt{T}(\widehat{\phi }_{p}-\phi _{0,p})^{2}+  \notag \\
&&\frac{1}{12}(\frac{\partial ^{4}m_{T}}{\partial \phi _{p}^{4}}(\phi
_{0}))^{\prime }WS_{T}T(\widehat{\phi }_{p}-\phi _{0,p})^{4}.  \label{r4til}
\end{eqnarray}

Define $\widetilde{W}_{T}=(\frac{\partial m_{T}}{\partial \phi _{1}^{\prime }%
}(\phi _{0})$\ $\frac{1}{2}\frac{\partial ^{2}m_{T}}{\partial \phi _{p}^{2}}%
(\phi _{0}))^{\prime }W(\frac{\partial m_{T}}{\partial \phi _{1}^{\prime }}%
(\phi _{0})$\ $\frac{1}{2}\frac{\partial ^{2}m_{T}}{\partial \phi _{p}^{2}}%
(\phi _{0}))$ and partition $\widetilde{W}_{T}$ as\linebreak $\left[ 
\begin{array}{cc}
\widetilde{W}_{T,1,1} & \widetilde{W}_{T,1,p} \\ 
\widetilde{W}_{T,p,1} & \widetilde{W}_{T,p,p}%
\end{array}%
\right] $ $\vspace{0.06in}$where $\widetilde{W}_{N,p,p}$ is a scalar and $%
\widetilde{W}_{N,1,p}=\widetilde{W}_{N,p,1}^{\prime }$ is a vector. If $%
Z_{T}<0,$ then $Tm_{T}^{\prime }(\widehat{\phi })W_{T}m_{T}(\widehat{\phi }%
)=T\underline{K}_{T}(\phi _{0})+T^{-1/2}R_{4,T}(\sqrt{T}(\widehat{\phi }%
_{p}-\phi _{0,p})^{2})+2(\widehat{\phi }_{p}-\phi _{0,p})^{2}\widetilde{R}%
_{1T}(\sqrt{T}(\widehat{\phi }_{p}-\phi _{0,p})^{2},S_{1,T})+o_{p}(T^{-1/2})$
is in fact minimized at%
\begin{equation}
\sqrt{T}(\widehat{\phi }_{p}-\phi
_{0,p})^{2}=V_{T}+T^{-1/2}Y_{T}+o_{p}(T^{-1/2})
\end{equation}%
where $V_{T}=-2\sqrt{T}Z_{T}\mathbf{1}(\sqrt{T}Z_{T}<0)/\sigma _{G}$ and 
\begin{eqnarray}
Y_{T} &\equiv &-\frac{1}{2}\widetilde{Q}_{T}^{-1}(T((\frac{\partial m_{T}}{%
\partial \phi _{p}}(\phi _{0})+\frac{\partial ^{2}m_{T}}{\partial \phi
_{p}\partial \phi _{1}^{\prime }}(\phi _{0})(\widehat{\phi }_{1}-\phi
_{0,1}))^{\prime }W(\frac{\partial m_{T}}{\partial \phi _{p}}(\phi _{0})+ 
\notag \\
&&\frac{\partial ^{2}m_{T}}{\partial \phi _{p}\partial \phi _{1}^{\prime }}%
(\phi _{0})(\widehat{\phi }_{1}-\phi _{0,1}))+\frac{1}{2}(\kappa
_{1,T}\ldots \kappa _{q,T})^{\prime }W\frac{\partial ^{2}m_{T}}{\partial
\phi _{p}^{2}}(\phi _{0}))+2\widetilde{R}_{1T}(V_{T},S_{1,T})+  \notag \\
&&\frac{2}{3}\sqrt{T}(\frac{\partial m_{T}}{\partial \phi _{p}}(\phi _{0})+%
\frac{\partial ^{2}m_{T}}{\partial \phi _{p}\partial \phi _{1}^{\prime }}%
(\phi _{0})(\widehat{\phi }_{1}-\phi _{0,1}))^{\prime }W\frac{\partial
^{3}m_{T}}{\partial \phi _{p}^{3}}(\phi _{0})V_{T}+  \notag \\
&&\frac{1}{12}(\frac{\partial ^{3}m_{T}}{\partial \phi _{p}^{3}}(\phi
_{0}))^{\prime }W\frac{\partial ^{3}m_{T}}{\partial \phi _{p}^{3}}(\phi
_{0})V_{T}^{2})  \label{yt}
\end{eqnarray}%
with%
\begin{equation*}
\widetilde{Q}_{T}=\widetilde{W}_{T,p,p}-\widetilde{W}_{T,1,p}^{\prime }%
\widetilde{W}_{T,1,1}^{-1}\widetilde{W}_{T,1,p}.
\end{equation*}

Let $\widetilde{W}=(D$\ $\frac{1}{2}G)^{\prime }W(D$\ $\frac{1}{2}G)$ and
partition $\widetilde{W}$ in the same way as $\widetilde{W}_{T}.$ Replacing $%
(\widehat{\phi }_{1}-\phi _{0,1})$ in (\ref{r1til}), (\ref{r4til}) and (\ref%
{yt}) by its expression from (\ref{g3}), we obtain%
\begin{eqnarray}
\widetilde{R}_{1T}(\sqrt{T}(\widehat{\phi }_{p}-\phi _{0,p})^{2},S_{T}) &=&%
\sqrt{T}(Z_{1T}+G_{1p}H(Z_{0T}+\frac{1}{2}G(\widehat{\phi }_{p}-\phi
_{0,p})^{2}))^{\prime }WS_{T}+  \notag \\
&&\frac{1}{3!}L^{\prime }WS_{T}\sqrt{T}(\widehat{\phi }_{p}-\phi
_{0,p})^{2}+o_{p}(1).  \label{xr1til}
\end{eqnarray}%
\begin{eqnarray}
\widetilde{R}_{4,T}(\sqrt{T}(\widehat{\phi }_{p}-\phi _{0,p})^{2},S_{T}) &=&%
\sqrt{T}(G_{1pp}H(Z_{0T}+\frac{1}{2}G(\widehat{\phi }_{p}-\phi
_{0,p})^{2}))^{\prime }WS_{T}\sqrt{T}(\widehat{\phi }_{p}-\phi _{0,p})^{2}+ 
\notag \\
&&\frac{1}{12}F^{\prime }WS_{T}T(\widehat{\phi }_{p}-\phi
_{0,p})^{4}+o_{p}(1).  \label{xr4til}
\end{eqnarray}%
and 
\begin{eqnarray}
Y_{T} &\equiv &-\frac{1}{2}\widetilde{Q}^{-1}(T((Z_{1T}+G_{1p}H(Z_{0T}+\frac{%
1}{2}G(\widehat{\phi }_{p}-\phi _{0,p})^{2}))^{\prime
}W(Z_{1T}+G_{1p}H(Z_{0T}+  \notag \\
&&\frac{1}{2}G(\widehat{\phi }_{p}-\phi _{0,p})^{2}))+\frac{1}{2}(\lambda
_{1,T}\ldots \lambda _{q,T})^{\prime }WG)+2\widetilde{R}_{1T}(V_{T},S_{1,T})+
\notag \\
&&\frac{2}{3}\sqrt{T}(Z_{1T}+G_{1p}H(Z_{0T}+\frac{1}{2}G(\widehat{\phi }%
_{p}-\phi _{0,p})^{2}))^{\prime }WLV_{T}+\frac{1}{12}L^{\prime
}WLV_{T}^{2})+o_{p}(1)\qquad  \label{xyt}
\end{eqnarray}%
with $\widetilde{Q}=\widetilde{W}_{p,p}-\widetilde{W}_{1,p}^{\prime }%
\widetilde{W}_{1,1}^{-1}\widetilde{W}_{1,p}.$

Let $\underline{R}_{2,T}(\phi _{0})=(\widehat{\phi }_{p}-\phi _{0,p})\times
2R_{2,T}$ with $R_{2,T}=R_{5,T}+T^{-3/2}\widetilde{R}_{1T}(\sqrt{T}(\widehat{%
\phi }_{p}-\phi _{0,p})^{2},S_{2,T})+\frac{1}{2}T^{-3/2}\widetilde{R}_{4,T}(%
\sqrt{T}(\widehat{\phi }_{p}-\phi _{0,p})^{2},S_{1,T}).$ At the minimum of $%
Q_{T}(\phi )$, we expect $\underline{R}_{2,T}(\phi _{0})$ to be negative,
i.e., $(\widehat{\phi }_{p}-\phi _{0,p})$ and $R_{2,T}$ have opposite sign.
Hence,%
\begin{equation}
T^{1/4}(\widehat{\phi }_{p}-\phi _{0,p})=(-1)^{B_{2,T}}T^{1/4}|\widehat{\phi 
}_{p}-\phi _{0,p}|,  \label{h4}
\end{equation}%
with $B_{2,T}=\mathbf{1(}T^{6/4}R_{2,T}\geq 0\mathbf{)}$. After replacing $(%
\widehat{\phi }_{p}-\phi _{0,p})^{2}$ in $R_{2,T}$ by its expression from (%
\ref{g6}), we can see, using the continuous mapping theorem, that $%
T^{6/4}R_{2,T}$ converges in distribution towards, say, $\mathcal{R}_{2}.$
The formula for $\mathcal{R}_{2}$ is given by $R_{2,T}$ with $\sqrt{T}%
Z_{0T}, $ $\sqrt{T}Z_{1T},$ $S_{1,T},$ $S_{2,T}$ and\ powers of $\sqrt{T}(%
\widehat{\phi }_{p}-\phi _{0,p})^{2}$\ and $V_{T}$ replaced by $\mathcal{Z}%
_{0},$ $\mathcal{Z}_{1},$ $S_{1,},$ $S_{2}$ and powers of $V$ $(=-2\mathcal{Z%
}\mathbf{1}(\mathcal{Z}<0)/\sigma _{G}),$ respectively, and with the
unspecified $o_{p}(T^{-6/4})$ and $o_{p}(1)$ terms at the very end of (\ref%
{h3}), (\ref{xr1til}), (\ref{xr4til}) and (\ref{xyt}) omitted.

We actually have that $(\sqrt{T}Z_{0T},$ $\sqrt{T}Z_{1T},$ $T^{6/4}R_{2,T})$
converges in distribution towards\linebreak $(\mathcal{Z}_{0},$ $\mathcal{Z}%
_{1},$ $\mathcal{R}_{2})$. Furthermore, recall that we have assumed that $%
\Pr (\mathcal{R}_{2}=0)=0.$

Applying a version of Lemma 1 of DH (2018), we have $(\sqrt{T}Z_{0T},$ $%
\sqrt{T}Z_{1T},$ $(-1)^{B_{2,T}})\overset{d}{\rightarrow }(\mathcal{Z}_{0},$ 
$\mathcal{Z}_{1},$ $(-1)^{B_{2}})$, where $B_{2}=\mathbf{1}(\mathcal{R}%
_{2}\geq 0)$. Since $(\sqrt{T}(\widehat{\phi }_{1}-\phi _{0,1}),$ $T^{1/4}|%
\widehat{\phi }_{p}-\phi _{0,p}|,$ $(-1)^{B_{2,T}})=O_{p}(1)$, any
subsequence of the left hand side has a further subsequence that converges
in distribution. Using part (b) of Lemma 1, such a subsequence obeys $(\sqrt{%
T}(\widehat{\phi }_{1}-\phi _{0,1}),$ $T^{1/4}|\widehat{\phi }_{p}-\phi
_{0,p}|,$ $(-1)^{B_{2,T}})\overset{d}{\rightarrow }(H\mathcal{Z}_{0}+HGV/2,$ 
$\sqrt{V},$ $(-1)^{B_{2}})$. Since the limit distribution does not depend on
the subsequence, the whole sequence converges towards that limit. By the
continuous mapping theorem, we deduce that: $(\sqrt{T}(\widehat{\phi }%
_{1}-\phi _{0,1}),$ $T^{1/4}(\widehat{\phi }_{p}-\phi _{0,p}))\overset{d}{%
\rightarrow }(H\mathcal{Z}_{0}+HGV/2,$ $(-1)^{B_{2}}\sqrt{V})$. ${\tiny %
\blacksquare }$

\begin{lemma}
(a1) When $q=p,$ then $\Pr (\mathcal{R}_{1}=0|\mathcal{Z}<0\mathbf{)}=1.$

(a2) When $q>p,$ then $\Pr (\mathcal{R}_{1}=0\mathbf{)}=0.$

(b) When $q=p,$ then $\Pr (\mathcal{R}_{2}=0|\mathcal{Z}<0)>0$ iff $F=0,$ $%
G_{1pp}=0,$ $K_{k}=0$ for $k=1,2,\ldots ,q,$ and $S_{2}=0$.
\end{lemma}

\textbf{Proof of Lemma 2. }$\mathcal{R}_{1}$ is the sum of two terms with
the first term given in the first line of (\ref{g9}) and the second term
given in the second line of (\ref{g9}).

(a1) We will first show that$\vspace{-0.04in}$ 
\begin{equation*}
(\mathcal{Z}_{0}^{\prime }W^{1/2}P_{g}W^{1/2}\mathcal{Z}_{0}G^{\prime
}-G^{\prime }W^{1/2}P_{g}W^{1/2}\mathcal{Z}_{0}\mathcal{Z}_{0}^{\prime
})W^{1/2}M_{d}W^{1/2}\mathbf{1(}\mathcal{Z}<0\mathbf{)}=0
\end{equation*}%
so that the first term of $\mathcal{R}_{1}$ in (\ref{g9}) equals $0$ when $%
\mathcal{Z}<0.$

Recall that 
\begin{equation*}
P_{g}=M_{d}W^{1/2}G(G^{\prime }W^{1/2}M_{d}W^{1/2}G)^{-1}G^{\prime
}W^{1/2}M_{d}
\end{equation*}%
where $M_{d}=I-W^{1/2}D(D^{\prime }WD)^{-1}D^{\prime }W^{1/2}.$

Define$\vspace{-0.23in}$%
\begin{equation*}
Z_{1}=-2\mathcal{Z}/\sigma _{G}=-2(G^{\prime
}W^{1/2}M_{d}W^{1/2}G)^{-1}G^{\prime }W^{1/2}M_{d}W^{1/2}\mathcal{Z}_{0}.
\end{equation*}%
Note that $\mathbf{1(}\mathcal{Z}<0\mathbf{)=1(}Z_{1}>0)$ and $\sqrt{T}(%
\widehat{\phi }_{p}-\phi _{0,p})^{2}\rightarrow ^{d}V=Z_{1}\mathbf{1(}%
Z_{1}>0).$

Now$\vspace{-0.12in}$%
\begin{gather*}
(\mathcal{Z}_{0}^{\prime }W^{1/2}P_{g}W^{1/2}\mathcal{Z}_{0}G^{\prime
}-G^{\prime }W^{1/2}M_{d}W^{1/2}\mathcal{Z}_{0}\mathcal{Z}_{0}^{\prime
})W^{1/2}M_{d}W^{1/2}\mathbf{1(}\mathcal{Z}<0\mathbf{)}= \\
(\mathcal{Z}_{0}^{\prime }W^{1/2}M_{d}W^{1/2}G(G^{\prime
}W^{1/2}M_{d}W^{1/2}G)^{-1}G^{\prime }W^{1/2}M_{d}W^{1/2}\mathcal{Z}%
_{0}G^{\prime }- \\
G^{\prime }W^{1/2}M_{d}W^{1/2}\mathcal{Z}_{0}\mathcal{Z}_{0}^{\prime
})W^{1/2}M_{d}W^{1/2}\mathbf{1(}\mathcal{Z}<0\mathbf{)}= \\
(-\frac{1}{2}\mathcal{Z}_{0}^{\prime }W^{1/2}M_{d}W^{1/2}GZ_{1}G^{\prime
}-G^{\prime }W^{1/2}M_{d}W^{1/2}\mathcal{Z}_{0}\mathcal{Z}_{0}^{\prime
})W^{1/2}M_{d}W^{1/2}\mathbf{1(}\mathcal{Z}<0\mathbf{)}= \\
-\mathcal{Z}_{0}^{\prime }W^{1/2}M_{d}W^{1/2}G(\frac{1}{2}Z_{1}G^{\prime }+%
\mathcal{Z}_{0}^{\prime })W^{1/2}M_{d}W^{1/2}\mathbf{1(}\mathcal{Z}<0\mathbf{%
)}=0
\end{gather*}%
because $(\mathcal{Z}_{0}+\frac{1}{2}Z_{1}G)^{\prime }W^{1/2}M_{d}W^{1/2}%
\mathbf{1(}Z_{1}>0)=0.$ The latter can be shown as follows. If $q=p$ and $%
Z_{T}<0$, then upon replacing $(\widehat{\phi }_{1}-\phi _{0,1},$ $(\widehat{%
\phi }_{p}-\phi _{0,p})^{2})^{\prime }$ in (\ref{g4}) by\linebreak $-(D$ $%
\frac{1}{2}G)^{-1}m_{T}(\phi _{0})+o_{p}(T^{-1/2})$, we obtain $\underline{K}%
_{T}(\phi _{0})=o_{p}(T^{-1})$. Hence it follows from (\ref{g2}) that if $%
q=p $ and $Z_{T}<0,$ then $m_{T}^{\prime }(\widehat{\phi })Wm_{T}(\widehat{%
\phi })\leq $ $o_{p}(T^{-1})$ and hence $m_{T}(\widehat{\phi }%
)=o_{p}(T^{-1/2}).$ Using this result, that $\sqrt{T}Z_{T}\rightarrow ^{d}%
\mathcal{Z}$ and that $\sqrt{T}(\widehat{\phi }_{1}-\phi _{0,1})\rightarrow
^{d}\mathcal{Z}_{2}\equiv H\mathcal{Z}_{0}+HGV/2$, it follows from (\ref{f8}%
) that $(\mathcal{Z}_{0}+D\mathcal{Z}_{2}+\frac{1}{2}Z_{1}G)\mathbf{1(}%
Z_{1}>0)=0$.\footnote{%
I thank P. Dovonon and A. Hall for pointing out a flaw in an earlier version
of this argument.} Hence$\vspace{-0.12in}$ 
\begin{eqnarray*}
W^{-1/2}M_{d}W^{1/2}(\mathcal{Z}_{0}+\frac{1}{2}Z_{1}G)\mathbf{1(}Z_{1}
&>&0)= \\
(I-D(D^{\prime }WD)^{-1}D^{\prime }W)(\mathcal{Z}_{0}+D\mathcal{Z}_{2}+\frac{%
1}{2}Z_{1}G)\mathbf{1(}Z_{1} &>&0)=0.
\end{eqnarray*}

Next, we will show that $\mathcal{Z}_{0}^{\prime }W^{1/2}M_{dg}W^{1/2}%
\mathbf{1(}\mathcal{Z}<0\mathbf{)}=0$ so that the second term of $\mathcal{R}%
_{1}$ in (\ref{g9}) equals $0$ when $\mathcal{Z}<0$. Using $%
M_{dg}=M_{d}-P_{g}$ and the definition of $Z_{1}$, we obtain$\vspace{-0.07in}
$%
\begin{equation*}
\mathcal{Z}_{0}^{\prime }W^{1/2}M_{dg}W^{1/2}\mathbf{1(}\mathcal{Z}<0\mathbf{%
)}=(\mathcal{Z}_{0}^{\prime }+\frac{1}{2}Z_{1}G^{\prime })W^{1/2}M_{d}W^{1/2}%
\mathbf{1(}Z_{1}>0)=0
\end{equation*}%
again because $W^{-1/2}M_{d}W^{1/2}(\mathcal{Z}_{0}+\frac{1}{2}Z_{1}G)%
\mathbf{1(}Z_{1}>0)=0.$ We conclude that if $q=p,$ then $\mathcal{R}_{1}%
\mathbf{1(}\mathcal{Z}<0\mathbf{)}=0,$ which means that $\Pr (\mathcal{R}%
_{1}=0|\mathcal{Z}<0\mathbf{)}=1.$

We can actually show that $\mathcal{R}_{1}=0$ when $q=p$. It is easily
verified that the matrix $M_{dg}$ projects a vector on the orthogonal
complement of $(W^{1/2}D$ $W^{1/2}G).$ When $q=p,$ $Rank(D$ $G)=p=q.$
Moreover, $W$ has full rank. Hence when $q=p,$ then $M_{dg}=\mathbf{0}$ and $%
(\frac{1}{2}Z_{1}G^{\prime }+\mathcal{Z}_{0}^{\prime })W^{1/2}M_{d}W^{1/2}=%
\mathcal{Z}_{0}^{\prime }W^{1/2}M_{dg}W^{1/2}=0$ so that both terms of $%
\mathcal{R}_{1}$ equal $0.$

(a2) Note that only the second term of $\mathcal{R}_{1},$ that is, $\mathcal{%
Z}_{0}^{\prime }W^{1/2}M_{dg}W^{1/2}(\mathcal{Z}_{1}+G_{1p}H\mathcal{Z}_{0})$
depends on $\mathcal{Z}_{1}.$ It therefore suffices to show that $\Pr (%
\mathcal{Z}_{0}^{\prime }W^{1/2}M_{dg}W^{1/2}\mathcal{Z}_{1}=0)=0.$%
\linebreak The latter follows from $M_{dg}^{2}=M_{dg},$ $Rank(M_{dg})=q-p>0,$
and the fact that\linebreak $((W^{1/2}\mathcal{Z}_{0})^{\prime }$ $(W^{1/2}%
\mathcal{Z}_{1})^{\prime })^{\prime }$ has a continuous multivariate
(normal) distribution with mean zero and a covariance matrix of full rank.
We conclude that when $q>p$, $\Pr (\mathcal{R}_{1}=0\mathbf{)}=0.$

(b) When $q=p$ and $\mathcal{Z}<0,$ it is easily verified that $\mathcal{R}%
_{2}=0$ if and only if $F=0,$ $G_{1pp}=0,$ $K_{k}=0$ for $k=1,2,\ldots ,q,$
and $S_{2}=0$.\quad ${\tiny \blacksquare }\medskip $

\noindent \textbf{Optimal weight matrix for }$\widehat{\phi }_{1}:\smallskip 
$

Recall that $\sqrt{T}(\widehat{\phi }_{1}-\phi _{0,1})\overset{d}{%
\rightarrow }H(\mathcal{Z}_{0}+GV/2)$ with $H=-(D^{\prime }WD)^{-1}D^{\prime
}W.$ We assume that $\widehat{W}_{opt}$ or $\hat{\Psi}_{p}^{-1}$ has been
used to compute $\widehat{\phi }_{p}$. Let $\Psi _{1}=\Psi _{1}(W_{opt})$ or 
$\Psi _{1}=\Psi _{1}(\Psi _{p}^{-1}).$ If $W=\Psi _{1}^{-1},$ then the
limiting MSE of $\widehat{\phi }_{1}$ is $(D^{\prime }\Psi _{1}^{-1}D)^{-1}.$
The matrix $\Psi _{1}^{-1}$ is an optimal weight matrix for $\widehat{\phi }%
_{1}$ if for any $W$ the matrix $(H\Psi _{1}H^{\prime }-(D^{\prime }\Psi
_{1}^{-1}D)^{-1})$ is p.s.d.s. Now, $H\Psi _{1}H^{\prime }-(D^{\prime }\Psi
_{1}^{-1}D)^{-1}=H\Psi _{1}^{1/2}(I-\Psi _{1}^{-1/2}D(D^{\prime }\Psi
_{1}^{-1}D)^{-1}D^{\prime }\Psi _{1}^{-1/2})\Psi _{1}^{1/2}H^{\prime }$ is
p.s.d.s. because $I-\Psi _{1}^{-1/2}D(D^{\prime }\Psi
_{1}^{-1}D)^{-1}D^{\prime }\Psi _{1}^{-1/2}$ is idempotent.\quad ${\tiny %
\blacksquare }\medskip $

\noindent O\textbf{ptimal weight matrix for }$\widehat{\phi }_{p}:\smallskip 
$

Let $\widehat{\phi }_{p}^{\ast }$ and $\widehat{\phi }_{p}$ denote the GMM
estimators for $\phi _{p}$ that use, respectively, $\Psi _{p}^{-1}$ and an
arbitrary matrix $W$ as weight matrices. Recall that $\sqrt{T}(\widehat{\phi 
}_{p}-\phi _{0,p})^{2}\overset{d}{\rightarrow }V$ and that $\sigma
_{G}=G^{\prime }W^{1/2}M_{d}W^{1/2}G.$ Using the properties of a half-normal
distribution, the limiting MSE\ of $\widehat{\phi }_{p}$ can easily be shown
to be equal to $(G^{\prime }W^{1/2}M_{d}W^{1/2}\Psi
_{p}W^{1/2}M_{d}W^{1/2}G)^{1/2}/\sigma _{G}\times \sqrt{2/\pi }.$

Let $M_{p}=I-\Psi _{p}^{-1/2}G(G^{\prime }\Psi _{p}^{-1/2}M_{d}(\Psi
_{p}^{-1})\Psi _{p}^{-1/2}G)^{-1}G^{\prime }\Psi _{p}^{-1/2}.$ Then the
difference\linebreak between the squares of the limiting MSEs of $\widehat{%
\phi }_{p}$ and $\widehat{\phi }_{p}^{\ast }$ is proportional to\linebreak $%
\Delta =G^{\prime }W^{1/2}M_{d}W^{1/2}\Psi _{p}^{1/2}M_{p}\Psi
_{p}^{1/2}W^{1/2}M_{d}W^{1/2}G.$ We will show that $\Delta \geq 0$ for all $%
W.$

Let $u=\Psi _{p}^{-1/2}G$ and $v=\Psi _{p}^{-1/2}D$. Then $%
M_{p}=I-u(u^{\prime }u-u^{\prime }v(v^{\prime }v)^{-1}v^{\prime
}u)^{-1}u^{\prime }.$

When $q>p=1$ or $u^{\prime }v=0,$ then $M_{p}=I-u(u^{\prime
}u)^{-1}u^{\prime },$ which is idempotent, and hence $\Delta \geq 0$ for all 
$W.$ It follows that in this case $\Psi _{p}^{-1}$ is an optimal (in the
sense of limiting MSE minimising) weight matrix for\textbf{\ }$\widehat{\phi 
}_{p}$ and the limiting MSE of $\widehat{\phi }_{p}$ is given by $(G^{\prime
}\Psi _{1}^{-1}G)^{-1/2}\sqrt{2/\pi }.$

When $q>p>1$ and $u^{\prime }v\neq 0,$ then $M_{p}$ is no longer idempotent.
Let $\tilde{W}=\Psi _{p}^{1/2}W\Psi _{p}^{1/2}$ and $M_{\tilde{W}%
^{1/2}v}=I_{q}-P_{\tilde{W}^{1/2}v}=I_{q}-\tilde{W}^{1/2}v(v^{\prime }\tilde{%
W}v)^{-1}v^{\prime }\tilde{W}^{1/2}.$ It follows that $\Psi
_{p}^{1/2}W^{1/2}M_{d}W^{1/2}\Psi _{p}^{1/2}=\tilde{W}-\tilde{W}v(v^{\prime }%
\tilde{W}v)^{-1}v^{\prime }\tilde{W}=\tilde{W}^{1/2}M_{\tilde{W}^{1/2}v}%
\tilde{W}^{1/2}$ and that $\Delta =u^{\prime }\tilde{W}^{1/2}M_{\tilde{W}%
^{1/2}v}\tilde{W}^{1/2}M_{p}\tilde{W}^{1/2}M_{\tilde{W}^{1/2}v}\tilde{W}%
^{1/2}u.$

Let $M_{u}=I-u(u^{\prime }u)^{-1}u^{\prime },$ $P_{v}=v(v^{\prime
}v)^{-1}v^{\prime }$ and $M_{v}=I-P_{v}.$ Then $M_{p}=(u^{\prime
}uM_{u}-(u^{\prime }P_{v}u)I)/(u^{\prime }M_{v}u).$ Now, $M_{u}=C_{u}\Lambda
_{u}C_{u}^{\prime },$ where $\Lambda _{u}=diag(1,1,...,1,1,0)$ and $%
C_{u}C_{u}^{\prime }=C_{u}^{\prime }C_{u}=I.$ Hence $M_{p}=u^{\prime
}uC_{u}(\Lambda _{u}-(u^{\prime }P_{v}u/u^{\prime }u)I)C_{u}^{\prime
}/(u^{\prime }M_{v}u).$ Furthermore, $M_{u}u=0,$ $M_{u}=C_{u}\Lambda
_{u}C_{u}^{\prime }$ and $C_{u}^{\prime }C_{u}=I$ imply that $[C_{u}^{\prime
}u]_{k}=0$ for $k=1,2,...,q-1$ and $[C_{u}^{\prime }]_{q,.}=u^{\prime }/%
\sqrt{u^{\prime }u}.$ Finally, let $\lambda =1+u^{\prime }P_{v}u/(u^{\prime
}M_{v}u)=u^{\prime }u/(u^{\prime }M_{v}u)$ and $\Lambda
=diag(0,0,...,0,0,\lambda ).$ Then we obtain $M_{p}=C_{u}(I_{q}-\Lambda
)C_{u}^{\prime }$ and$\vspace{-0.1in}$%
\begin{gather*}
\Delta =u^{\prime }\tilde{W}^{1/2}M_{\tilde{W}^{1/2}v}\tilde{W}%
^{1/2}C_{u}(I_{q}-\Lambda )C_{u}^{\prime }\tilde{W}^{1/2}M_{\tilde{W}^{1/2}v}%
\tilde{W}^{1/2}u= \\
u^{\prime }\tilde{W}^{1/2}M_{\tilde{W}^{1/2}v}\tilde{W}^{1/2}\tilde{W}%
^{1/2}M_{\tilde{W}^{1/2}v}\tilde{W}^{1/2}u-\lambda ([C_{u}^{\prime }]_{q,.}%
\tilde{W}^{1/2}M_{\tilde{W}^{1/2}v}\tilde{W}^{1/2}u)^{2}= \\
u^{\prime }\tilde{W}^{1/2}M_{\tilde{W}^{1/2}v}\tilde{W}^{1/2}\tilde{W}%
^{1/2}M_{\tilde{W}^{1/2}v}\tilde{W}^{1/2}u-(u^{\prime }\tilde{W}^{1/2}M_{%
\tilde{W}^{1/2}v}\tilde{W}^{1/2}u)^{2}/(u^{\prime }M_{v}u).
\end{gather*}

To prove $\Delta \geq 0$ for all $W$, it is sufficient to show that $%
(u^{\prime }M_{v}u)\Delta \geq 0$ for all $\tilde{W}.$ Using $%
u=P_{v}u+M_{v}u,$ we obtain $M_{\tilde{W}^{1/2}v}\tilde{W}^{1/2}u=M_{\tilde{W%
}^{1/2}v}\tilde{W}^{1/2}M_{v}u$ and $(u^{\prime }M_{v}u)\Delta =(u^{\prime }%
\tilde{W}^{1/2}M_{\tilde{W}^{1/2}v}\tilde{W}^{1/2}\tilde{W}^{1/2}M_{\tilde{W}%
^{1/2}v}\tilde{W}^{1/2}u)(u^{\prime }M_{v}M_{v}u)-(u^{\prime }\tilde{W}%
^{1/2}M_{\tilde{W}^{1/2}v}\tilde{W}^{1/2}M_{v}u)^{2}.$ It follows from
applying the Cauchy-Schwarz inequality that $(u^{\prime }M_{v}u)\Delta \geq
0 $ for all $\tilde{W}$ and hence $\Delta \geq 0$ for all $W.$ We conclude
that when $q>p>1$ and $u^{\prime }v\neq 0,$ then $\Psi _{p}^{-1}$ is still
the optimal (in the sense of limiting MSE minimising) weight matrix for%
\textbf{\ }$\widehat{\phi }_{p}$ and that the limiting MSE of $\widehat{\phi 
}_{p}$ in this case is given by $(G^{\prime }\Psi _{p}^{-1/2}M_{d}(\Psi
_{p}^{-1})\Psi _{p}^{-1/2}G)^{-1/2}\sqrt{2/\pi }.$\quad ${\tiny \blacksquare 
}$

\end{document}